\documentclass[11pt]{article}

\usepackage{times}
\usepackage{titlesec}

\usepackage{ifthen}

\ifthenelse{\isundefined{\GenerateShortVersion}}
{
\def\LongVersion{}
\def\LongVersionEnd{}
\long\def\ShortVersion#1\ShortVersionEnd{}
}
{
\def\ShortVersion{}
\def\ShortVersionEnd{}
\long\def\LongVersion#1\LongVersionEnd{}
}

\newcommand{\Comment}[1]{\ignorespaces}

\usepackage{algorithmic}
\usepackage{algorithm}

\usepackage{amsmath, amssymb}

\usepackage{amsthm}

\usepackage{ifpdf}

\usepackage{caption}
\usepackage{subcaption}

\usepackage{verbatim}

\ifpdf
\usepackage[pdftex]{graphicx}
\usepackage[pdftex,
  ocgcolorlinks,
  linkcolor=blue,
  filecolor=blue,
  citecolor=blue,
  urlcolor=blue]{hyperref}
\else
\usepackage[dvips]{graphicx}
\fi

\ifthenelse{\isundefined{\NoGenerateLetterSize}}
{
\ifpdf
\setlength{\pdfpagewidth}{8.5in}
\setlength{\pdfpageheight}{11in}
\else

\fi
}
{}
\LongVersion 
\LongVersionEnd 

\ShortVersion 
\titlespacing*{\section}
{0pt}{1.55ex}{1.25ex}
\titlespacing*{\subsection}
{0pt}{1.55ex}{1ex}

\ShortVersionEnd

\ShortVersion
\renewenvironment{proof}[0]
{\comment\ignorespaces}
{\endcomment\ignorespacesafterend}
\ShortVersionEnd

\ShortVersion
\newcommand{\ThmOnlyInLong}{\addtocounter{theorem}{1}}
\ShortVersionEnd
\LongVersion
\newcommand{\ThmOnlyInLong}{}
\LongVersionEnd

\newtheorem{theorem}{Theorem}[section]
\newtheorem{lemma}[theorem]{Lemma}

\newtheorem{proposition}[theorem]{Proposition}
\newtheorem{claim}[theorem]{Claim}
\newtheorem{definition}[theorem]{Definition}
\newtheorem{example}[theorem]{Example}

\newtheorem*{definition*}{Definition}
\theoremstyle{plain}

\usepackage{latexsym,amsmath,amssymb}
\usepackage{amsthm}
\usepackage{color,soul}
\usepackage{ifthen}
\usepackage{xspace}
\usepackage{paralist}
\usepackage{xstring}
\usepackage{graphicx}

\usepackage{fullpage}

\newcommand{\eps}{\epsilon}

\newcommand{\IC}{\mathsf{IC}}

\newcommand{\CC}{\mathsf{CC}}

\newcommand{\PIC}{\mathsf{PIC}}
\newcommand{\MIC}{\mathsf{MIC}}
\newcommand{\SIC}{\mathsf{SMIC}}

\newcommand{\AND}{\mathsf{AND}}
\newcommand{\XOR}{\mathsf{Par}}
\newcommand{\Disj}{\mathsf{Disj}}
\newcommand{\Ber}{\text{\textbf{Ber}}}
\newcommand{\calA}{\mathcal{A}}
\newcommand{\calB}{\mathcal{B}}
\newcommand{\calX}{\mathcal{X}}
\newcommand{\calY}{\mathcal{Y}}

\newcommand{\calR}{\mathcal{R}}

\newcommand{\calT}{\mathcal{T}}
\newcommand{\calI}{\mathcal{I}}
\newcommand{\calJ}{\mathcal{J}}
\newcommand{\calH}{\mathcal{H}}

\newcommand{\paragraphA}[1]{\vspace{0.125cm}\noindent {\bf #1}}

\DeclareMathOperator*{\Exp}{\mathbb{E}}

\newcommand{\Ni}[2]{\ensuremath{[\![#1,#2]\!]}}

\newcommand{\M}[3]{
    \IfEqCase{#3}{
        {r}{T_{#1}^{\text{\makebox[0mm]{\,\,$\overleftarrow{\color{white}l}$}}#2}}
        {s}{T_{#1}^{\text{\makebox[0mm]{\,\,$\overrightarrow{\color{white}l}$}}#2}}
    }[\PackageError{M}{Undefined option to M: #2}{}]
}
\newcommand{\T}[3]{
    \IfEqCase{#3}{
        {r}{T_{#1}^{<\text{\makebox[0mm]{\,\,$\overleftarrow{\color{white}l}$}}#2}}
        {s}{T_{#1}^{<\text{\makebox[0mm]{\,\,$\overrightarrow{\color{white}l}$}}#2}}
    }[\PackageError{T}{Undefined option to T: #2}{}]
}

\usepackage{ifpdf}

\newcommand{\adifuture}[1]{}

\title{A New Approach to Multi-Party  Peer-to-Peer Communication Complexity\footnote{A preliminary version of this paper 
appeared in the proceedings of ITCS 2019~\cite{RosenU19}.}
}

\author{Adi Ros\'{e}n\thanks{CNRS and Universit\'{e} de Paris, 75205 Paris, France,
email: {\tt adiro@irif.fr}. Research  supported in part by 
 ANR project RDAM.}
 \and
 Florent Urrutia\thanks{Universit\'{e} de Paris, 75205 Paris, France, email: {\tt urrutia@irif.fr}.
 Research supported in part by ERC QCC and by ANR project RDAM.}
 }
 
 \date{}

\begin{document}

\begin{titlepage}

\maketitle

\thispagestyle{empty}

\maketitle

\thispagestyle{empty}

\begin{abstract}

We introduce new  models and new information theoretic measures for the study of  communication complexity 
in the natural  peer-to-peer, multi-party, number-in-hand setting. We prove a number of properties of
our new models and measures, and then, in order to exemplify their effectiveness,
we use them to prove two lower bounds. The more elaborate one is a tight lower bound 
of $\Omega(kn)$ on the multi-party peer-to-peer randomized communication complexity of the $k$-player, $n$-bit 
function Disjointness, $\Disj_k^n$. The other one is a tight lower bound 
of $\Omega(kn)$ on the multi-party peer-to-peer  randomized communication complexity of the $k$-player, $n$-bit bitwise parity
function, $\XOR_k^n$.  Both lower bounds hold when ${n=\Omega(k)}$. The lower bound for $\Disj_k^n$
 improves over the lower bound that can 
be inferred from the result 
of Braverman et al.~(FOCS 2013),  which was proved in the coordinator model and can yield a lower bound of 
$\Omega(kn/\log k)$ in the peer-to-peer model. 

 To the best of our knowledge, our lower bounds  are the first tight 
(non-trivial)
lower bounds on communication complexity in the natural {\em peer-to-peer} multi-party setting.

In addition to the above results for communication complexity, we also prove, using the same tools,
 an $\Omega(n)$ lower bound on the number of random bits necessary for the (information theoretic) private 
 computation of the function $\Disj_k^n$.

\end{abstract}

\end{titlepage}

\section{Introduction}

Communication complexity, first introduced by Yao~\cite{Yao82}, has become a major topic of research in 
Theoretical Computer Science, both for its own sake, and as a tool which has yielded important results 
(mostly lower bounds)
in  various theoretical computer science fields  such as circuit complexity, streaming algorithms, or 
data structures (e.g.,~\cite{KN, MNSW, GG, SHKKNPPW, FHW}). 
Communication complexity is a measure for the amount of  communication needed in order to solve a problem whose
input is  distributed among several players. 
The two-party case, where two players, usually called Alice and Bob, cooperate in order to compute a function of their 
respective inputs, 
has been widely studied with many important results; yet major questions in this area are still open today 
(e.g., the log-rank 
conjecture, see~\cite{KN}). The multi-party case, where $k\geq 3$ players cooperate in order to compute a function
 of their inputs, is much less understood.
 
A number of variants have been proposed in the literature to extend the two-party setting into the multi-party one.
In this paper we consider the  more natural  {\em number-in-hand} 
(NIH) setting, where  each 
player has its own input, as opposed to the so-called {\em number-on-forehead} (NOF) setting, where each player
knows all pieces of the input except one, its own.
Moreover, also the communication structure between the players   in the 
 multi-party setting was considered in the literature under a number of variants.  
 For example, in the {\em blackboard} (or {\em broadcast}) model the communication  between the players 
is achieved by each player writing, in turn, a message on the board, to be read by all other players. 
In the {\em coordinator} model, introduced in \cite{DF}, there is an additional entity, the coordinator, and 
all players communicate 
back and forth only with the coordinator.
The most natural setting is, however,  the {\em peer-to-peer message-passing} model, where each 
pair of players is connected by a 
communication link, and  each player can send a separate message to any other player. This latter 
setting has been studied,
in the context of communication complexity, even less than the other multi-party settings, probably due to the difficulty
 in tracking the distributed communication patterns that occur during a run of a protocol in that setting.
  This setting is, 
 however, not only the most natural one, and the one that occurs the most in real systems, but is the 
 setting studied widely 
 in the distributed algorithms and distributed computation communities, for complexity measures which are usually
 other than communication complexity.
 
In the present paper we attempt to fill this gap in the study of peer-to-peer communication complexity,
 and, further, to create a more solid bridge between the research field of communication complexity 
and the research field of distributed computation. We propose a computation model, together with an
information theoretic  complexity measure, 
for the analysis
of the communication complexity of protocols in the asynchronous multi-party peer-to-peer (number-in-hand) setting. 
We 
argue that our model is, on the 
one hand, only a slight restriction over the  asynchronous model usually used in the distributed computation literature, and, on the 
other hand, 
stronger than the models that have been previously suggested in order to study communication complexity in 
the peer-to-peer setting common in the
distributed computation literature (e.g.,~\cite{DF,WZ2}).
Furthermore, our model lends itself to the analysis of 
communication complexity,
most notably using information theoretic tools.

Indeed, after defining  our model and our information theoretic measure, that we call {
\em Multi-party Information Cost} ($\MIC$), 
we prove a number of properties of that measure, and then prove  a number
of fundamental properties of protocols in our model. 
We then exemplify the effectiveness of our model and information theoretic measure
 by proving two tight lower bounds. The more elaborate one is 
 a tight lower bound of $\Omega(kn)$, when ${n=\Omega(k)}$, on the peer-to-peer 
randomized communication complexity
of the function {\tt set-disjointness} ($\Disj_k^n$).
This function is a basic, important 
function, which has been the subject of a large number of studies in communication complexity, 
and is  often seen as a test for our ability to give lower bounds in a given model (cf.~\cite{CP}). We note that the 
communication  complexity  of Disjointness
in the two-party case is well understood \cite{KS,Raz,BJKS,Bra,BGPW}. From 
a quantitative point of view, our result for peer-to-peer multi-party Disjointness improves by a $\log k$ 
factor the lower bound that could be
deduced for the peer-to-peer model from the lower bound on the communication complexity of Disjointness in the 
coordinator model~\cite{BEOPV}. The second lower bound that we prove is a tight lower bound of 
$\Omega(kn)$, when ${n=\Omega(k)}$, on the peer-to-peer 
randomized communication complexity of the bitwise parity function $\XOR_k^n$. 
Both our lower bounds are obtained by giving a lower bound on the $\MIC$ of the function at hand, which yields the
 lower bound on the communication complexity of that function. We believe that our lower bounds are the
first tight 
(non-trivial)
lower bound on communication complexity in a peer-to-peer multi-party setting.\footnote{Lower bounds in a seemingly peer-to-peer 
setting were given in~\cite{WZ2}. However, in the model of that paper, the communication pattern is determined 
by an external view 
of the transcript, which makes the model equivalent 
 to the coordinator model.}

It is important to note that, to the best of our knowledge, there is no known method to obtain tight lower bounds 
on multi-party communication complexity in a peer-to-peer setting via lower bounds in other known multi-party
settings.
Lower bounds obtained in the coordinator model can be transferred to the  peer-to-peer model at the cost of 
a $\log k$ factor, where $k$ is the number of players, because any peer-to-peer protocol can be simulated in the 
coordinator model by having the players attach to every message the identity of the destination of that 
message~\cite{PVZ,EOPV}. 
The loss of this factor in the lower bounds is unavoidable when the communication protocols can exploit  
a flexible communication 
pattern, since there are examples of functions where this factor  in the communication complexity is necessary, 
while others, e.g., the parity function of single-bit inputs,  have the same communication 
complexity in the coordinator and 
peer-to-peer settings 
 (see a more detailed
discussion on this point in Section~\ref{sec:comparison_models}).
Therefore, one cannot prove tight  
lower bounds in the peer-to-peer setting by proving corresponding results in the
 coordinator model.
Note that flexible communication 
configurations arise naturally for mobile communicating devices, for example, when
these devices 
exchange information with the nearby devices.
Constructions based on the pointer jumping problem  also seem to be harder in the coordinator model, 
as solving the problem usually requires exchanging information in a specific order determined by the inputs of the 
players.  
It  is thus important to develop lower bound techniques which apply directly 
 in the peer-to-peer model, as we 
do in the present paper. Information theoretic tools seem, as we show, 
 most suitable for this task. 

\paragraphA{Information theoretic complexity measures.}
As indicated above, our work makes  use of information theoretic tools.
Based on information theory, developed by Shannon~\cite{Sha}, {\em Information Complexity} (IC), originally defined 
in \cite{Bar-YehudaCKO93,CSWY}, is a powerful tool for the study of two-party communication protocols. Information 
complexity is a measure of how much {\em information}, about each other's input, the players must learn 
during the course of the protocol, if that protocol must compute the function correctly. Since IC can be 
shown to provide a lower bound on the communication complexity,  this measure
has proven to be a strong and useful tool for  obtaining lower bounds on two-party communication complexity
in a sequence of papers (e.g.,~\cite{BJKS,BBCR,BR,Bra}).
However, information complexity cannot be extended in a  straightforward manner to the multi-party setting. 
This is because with three players or more, any function can be computed privately (cf.~\cite{BGW,CCD}), i.e., in a 
way such 
that the players learn nothing but the value of the function to compute. This implies that the information 
complexity of any function is too low to provide a meaningful lower bound on the communication complexity in the 
natural peer-to-peer multi-party setting.  Therefore, before the present paper, 
information complexity and its variants have been used to obtain lower bounds
on multi-party communication complexity only in settings which do not allow for private protocols (and most 
notably not in the natural peer-to-peer setting), with the single exception of \cite{KRU}.
For example, a number of lower bounds have been obtained
 via information complexity 
for a promise version of  set-disjointness in the broadcast model~\cite{BJKS, CKS, Gro}  (also cf.~\cite{Jay}), and 
external information complexity was used in \cite{BO} for a lower bound on the general disjointness function, also 
in the broadcast model. In the coordinator model, lower bounds on the communication complexity of
set-disjointness were given
 via variants of information complexity~\cite{BEOPV}. The latter result was extended in \cite{CM} to the 
 function \textit{Tribes}. A notion of external information cost in the coordinator model was introduced in~\cite{HRVZ} to study 
 maximum matching in a distributed setting.
We  note that the study of communication complexity in number-in-hand multi-party settings via
 techniques other than those based on information theory is limited to very few papers. One such example is the 
 technique of \textit{symmetrization} that was introduced for the coordinator model in \cite{PVZ}, and 
was shown to be useful to  study functions such as the bitwise $\AND$. 
That technique was further developed along with other
reduction techniques in \cite{WZ,WZ2,WZ3}. Another example is the notion of strong fooling sets, 
introduced in~\cite{CK} to study deterministic communication complexity of discreet protocols, also defined in~\cite{CK}.

\paragraphA{Private computation.}
It is well known that in the multi-party number-in-hand peer-to-peer setting, unlike in the 
 two-party case, {\em any} function can be privately computed \cite{BGW,CCD}.
The model that we define in the present paper does allow for (information theoretic) private 
computation of any function~\cite{BGW,CCD,AsharovL17}.
 The minimum amount of private randomness needed 
 in order to compute privately a given function is often  referred to  in this context as the \textit{randomness complexity} of 
 that function. Randomness complexity  (in private computation) is of interest because true randomness is  considered  a 
 costly resource, and since  randomness complexity in private computation 
 has been shown to be related to other complexity 
 measures, such as the circuit size of the function or its sensitivity. For example,
 it has been shown~\cite{KOR} that a boolean function $f$ has a linear size circuit if and only if $f$ has constant 
 randomness complexity. A small number of
  works \cite{BSPV,KM,GR,KRU} prove lower bounds on the randomness complexity of the parity function. 
  The parity and other modulo-sum functions are, to the best of our knowledge, 
 the only functions for which randomness complexity lower bounds are known. Using the information theoretic 
 results that we obtain in 
 the present paper for the set-disjointness function, we are able to give a lower bound of $\Omega(n)$
 on the randomness complexity 
 of $\Disj_k^n$. The significance of this result lies in that it is the first such lower bound that grows with the size of the
input (which is $kn$), while the  output  remains a single bit, contrary to the sum 
function (see~\cite{BSPV}) or the bitwise parity function (see~\cite{KRU}).

\subsection{Our techniques and contributions}\label{subsec:ourapp}

Our contribution in the present paper is  twofold.

 First, on the conceptual, modeling and definitions side we 
lay the foundations for proving  lower bounds on (randomized) communication complexity  in the natural peer-to-peer multi-party 
setting. Specifically,  we
propose a model  that, on the one hand,  is a very natural peer-to-peer model,
 and very close to the model used in the distributed 
computation literature, and, at the same time, does have properties that allow one to  analyze  protocols
in terms of their information complexity and communication complexity. 
While at first sight the elaboration
of such model  does not seem to be a difficult task, many technical, as well as fundamental, issues render this
task non-trivial. For example, one would like to define a notion of ``transcript'' that would guarantee 
both a relation between
the length of the transcript and the communication complexity, and at the same time will contain
 {\em all} the information
that the players get and use while running the protocol.  The difficulty in elaborating such model
  may be the reason for which, prior to the present paper, 
hardly any work studied  
communication complexity directly in a  {\em peer-to-peer}, multi-party setting (cf.~\cite{EOPV}), leaving 
the field with only the results that can be inferred from other models, hence suffering 
 the appropriate loss in the obtained bounds. We 
propose our model  (see 
Section~\ref{subsec model})
and prove  a number of fundamental properties that allow one to analyze protocols in that model (see 
Section~\ref{sec:multi-party_protocols}), as well as prove the accurate relationship between the
entropy of the transcript and the communication complexity of the protocol (Proposition~\ref{prop:entropy_and_cc}).

We then define our new information theoretic measure, that we call ``Multi-party Information Cost'' ($\MIC$), 
intended 
to be applied to  peer-to-peer multi-party protocols, and 
prove that it provides, for any (possibly randomized) protocol, a lower bound on the communication complexity of that protocol
(Lemma~\ref{lem:CC>=MIC}). We further 
show that  $\MIC$  has certain properties such as
a certain direct-sum property (Theorem~\ref{thm:dsmic}). We thus introduce a framework as well as tools for
proving lower bounds on  communication complexity in a peer-to-peer multi-party setting.

Second, we exemplify the effectiveness of our conceptual contributions by proving, using the new tools that 
we define,
two tight lower bounds on the randomized communication complexity of certain functions in the peer-to-peer multi-party setting. 
Both these lower bounds are  proved by giving a lower bound on the Multi-party Information Complexity  
of the function at hand. The more elaborate lower bound is 
a tight lower bound of $\Omega(nk)$
on the randomized communication complexity of the function $\Disj_k^n$
(under the condition that $n=\Omega(k)$). The function Disjointness is  a well studied 
function in communication complexity and is often seen as a test-case
of one's ability to give lower bounds in a given model (cf.~\cite{CP}).
 While the general structure of the proof of this
lower bound 
does have similarities
 to the proof of a lower bound for Disjointness in the coordinator model~\cite{BEOPV},\footnote{The 
lower bound in~\cite{BEOPV} would yield an $\Omega(\frac{1}{\log k} \cdot nk)$ lower bound in the peer-to-peer setting.} 
we do, even in the parts that bear similarities, have 
 to overcome a number of technical 
 difficulties that require new ideas and
new proofs. For example,  
the very basic {\em rectangularity} property of communication protocols is, in the multi-party (peer-to-peer) 
setting, very sensitive to the details of 
the definition of the model and the notion of a transcript. We therefore need first to give a proof of 
 this property in the peer-to-peer model 
(Lemma~\ref{le:det_rectangularity} and Lemma~\ref{lem:rectang}).
We then  use a distribution of the input which is a
modification over the distributions used in~\cite{BEOPV,CM}
(see Section~\ref{sec:AND}). 
Our proof
proceeds, as in~\cite{BEOPV}, by proving a lower bound for
the function $\AND$,  on a certain information theoretic measure that, in our proof,
 is called $\SIC$ (for Switched Multi-party Information Cost), and then, by using a direct-sum-like lemma,
  to infer a lower bound 
on $\SIC$  for Disjointness (we note that $\SIC$ is  an adaptation to the peer-to-peer model of a 
similar measure used in~\cite{BEOPV}).
However, the lack of  a ``coordinator'' in a peer-to-peer setting 
necessitates a definition of
a more elaborate reduction protocol, and a  more complicated proof for the direct-sum argument, 
inspired  by classic secret-sharing primitives. 
See 
Lemma~\ref{lem:DSSIC} for our construction and proof. 
We then show that $\SIC$ provides a lower bound on  $\MIC$, which yields our lower bound on the
 communication complexity of Disjointness. 
 
We further give a tight lower bound of $\Omega(nk)$ on the randomized communication complexity 
of the function $\XOR_k^n$ 
(bitwise parity)
in the peer-to-peer multi-party setting (under the condition that $n = \Omega(k)$). This proof proceeds by first
giving a lower bound on $\MIC$ for the parity function $\XOR_k^1$, and then using a direct-sum property of $\MIC$
to get a lower bound on $\MIC$ for $\XOR_k^n$. The latter yields the lower bound of $\Omega(nk)$ on the 
communication complexity  of  $\XOR_k^n$.

To the best of our knowledge, our lower bounds are the first tight (non-trivial) lower bound on communication complexity
in a peer-to-peer multi-party setting.

In addition to our results on communication complexity, we
 analyze the number of random bits necessary for  private computations~\cite{BGW,CCD}, making use of 
the model, tools
and techniques we develop in the present paper. It has been  shown~\cite{KRU} that the \textit{public information cost} 
(defined also in~\cite{KRU}) can be used to derive a lower bound on the randomness complexity of private 
computations. In the present  paper we  give a lower bound on  the public information cost of any synchronous protocol
computing the Disjointness function by relating it to its Switched Multi-party Information Cost, which yields the lower bound
on the randomness complexity of Disjointness.

\paragraphA{Organization.}
\LongVersion
The appendix contains a short review of information theoretic notions that we use in the 
present paper.
We start the paper, in Section~\ref{sec:Model}, by introducing our  model and by comparing it to other models.
 In Section~\ref{sec:Prelim} we define our new information theoretic measure, $\MIC$, and prove some
 of its properties, and then prove a number of fundamental properties of protocols in our peer-to-peer model.
 In Section~\ref{sec:parity} we give the lower bound for the bitwise parity function. In
Section~\ref{sec:AND} we prove a lower bound on the switched multi-party information cost of the function $\AND_k$,
and in  Section~\ref{sec:Disj},  we prove, using the results of  Section~\ref{sec:AND}, the lower bound on 
the communication complexity of the disjointness function $\Disj_k^n$. 
In Section \ref{sec:Rand}, we show how 
to apply our  information theoretic lower bounds in order 
to give  a lower bound on the number of random bits necessary for
 the private computation of the function $\Disj_k^n$. 
  Last, in Section~\ref{sec:conclusions} we discuss some open questions. 
\LongVersionEnd
\ShortVersion
Due to space limitation {\em all} proofs are deferred to the  full version of the paper. 
Section~\ref{sec:Model}  introduces our  model.
 In Section~\ref{sec:Prelim} we define our new information theoretic measure, $\MIC$,  give some
 of its properties, and  give a number of fundamental properties of protocols in our model.
 In Section~\ref{sec:parity} we give the lower bound for the bitwise parity function. In
Section~\ref{sec:AND} we prove a lower bound on the switched multi-party information cost of $\AND_k$,
and in  Section~\ref{sec:Disj},  we prove, using the results of  Section~\ref{sec:AND}, the lower bound on 
the communication complexity of $\Disj_k^n$. 
In Section \ref{sec:Rand} we  
 apply our  information theoretic  lower bounds in order 
to give  a lower bound on the number of random bits necessary for
 the private computation of  $\Disj_k^n$. 
 Last, in Section~\ref{sec:conclusions} we discuss some open questions. 
 \ShortVersionEnd

\section{Multi-party communication protocols}\label{sec:Model}

We start with our model, and, to this end, give a number of notations.

\paragraphA{Notations.}
We denote by $k$ the number of players. We often use $n$ to denote the size (in bits) of the input to each player.
Calligraphic letters will be used to denote sets. Upper case letters will be used to denote random variables, and given
 two random variables $A$ and $B$, we will denote by $AB$ the joint random variable $(A,B)$. Given
  a string (of bits) $s$, $|s|$ denotes the length of $s$. Using parentheses we denote an ordered set (family) of items, 
  e.g., $(Y_i)$. Given a family $(Y_i)$,  $Y_{-i}$ denotes the sub-family which is the family $(Y_i)$ {\em without} 
  the element $Y_i$. The letter $X$ will usually denote the input to the players, and we thus use the shortened 
  notation $X$ for $(X_i)$, i.e.,  the input to all players. A protocol will usually be denoted by $\pi$.

We now define a natural communication model which is a slight restriction of the 
general asynchronous 
peer-to-peer model.
The restriction of our model compared to the general asynchronous 
peer-to-peer model  is that for a given player at a given time, the set of players from which that player waits for 
a message before sending any message of its own
 is determined by that player's own local view, i.e., from that player's input and the messages it has read so far,
as well as its private randomness, and the public randomness.
This allows us to define information theoretic tools that pertain to the transcripts of the protocols, and at 
the same time to use these tools as lower bounds for communication complexity. This restriction however 
does not exclude the existence of private protocols, as other special cases of the general asynchronous  model do.  
We observe  that practically all multi-party protocols in the 
literature are implicitly defined in our model, and that without such restriction, 
 one bit of communication
can bring $\log k$ bits of information, because not only the content of the message, but also the identity of the
sender may reveal information. 
To exemplify why  the general asynchronous  model is problematic consider the following simple 
example (that we borrow from  our work in~\cite{KRU}).
\begin{example}
\label{ex:problematic_protocol}
There are $4$ players $A$, $B$ and $C$, $D$. The protocol allows $A$ to transmit 
to $B$ its input bit $x$. But all messages sent in the protocol are the bit $0$, and the 
protocol generates only a single 
transcript over all possible inputs. The protocol works as follows:

{\bf A}: If $x=0$ send $0$ to $C$; after receiving $0$ from $C$, send $0$ to $D$.

\hspace{0.35cm} If $x=1$ send $0$ to $D$; after receiving $0$ from $D$, send $0$ to $C$

{\bf B}: After receiving $0$ from a player, send $0$ back to that player.

{\bf C,D}: After receiving $0$ from $A$ send $0$ to $B$. After receiving $0$ from $B$ send $0$ to $A$.

\noindent It is easy to see that $B$ learns the value of $x$ from the order of the messages it gets.
\end{example}

In what follows we formally define our model, compare it 
to the general one and to other restricted ones, and explain the usefulness and logic of our specific model. 

\subsection{Definition of the model}\label{subsec model}
We work in a \textit{multi-party, number-in-hand, peer-to-peer} setting. Each player $1\leq i \leq k$
 has unbounded local computation
 power and, in addition to its input $X_i$, has access to a source of private randomness $R_i$. We will use the 
notation $R$ 
for $(R_i)$, i.e., the private randomness of all players. A source of public randomness 
$R^p$ is also available to all players. We will call a protocol with no private randomness a public-coins protocol.
The system consists of $k$ players and a family of $k$ 
functions $f = (f_i)_{i\in\Ni{1}{k}}$, with $\forall~i \in\Ni{1}{k},~ f_i: \Pi_{\ell=1}^{k}\calX_\ell \rightarrow \calY_i$,
 where $\calX_\ell$ denotes the set of possible inputs of player $\ell$, and  ${\calY_i}$  denotes the set of possible 
 outputs of player $i$.
The players are given some input $x = (x_i)\in \Pi_{i=1}^{k}\calX_i$, and for every $i$, player $i$ has to 
compute $f_i(x)$. 

We define the communication model as follows, which is the asynchronous setting, with some restrictions.
To make the discussion simpler we assume a global time which is {\em unknown} to the players.
Every pair of players is connected by a bidirectional communication link that allows them to send
messages to each other. 
There is no bound on the delivery time (i.e., when the message arrives to its destination node) of a message, 
but every message
is delivered in finite time, and the communication link maintains FIFO order in each of the two directions. 
Messages that arrive to the head of the link at the destination node of that link are buffered until they are read by that node.
Given a specific time we define the {\em view} of player $i$ as the input of this player, 
$X_i$, its private randomness, $R_i$,  the public randomness, $R^p$, and the messages read so far by player $i$. 
After the protocol has started,The protocol of each player $i$ runs in  {\em local rounds}.
In each round, player $i$ sends messages to some subset of the other players. The identity of these players, as well as
the content of these messages, depend on the current view of player $i$.  The player also decides whether
it should stop, and output (or ``return'') the result of the function $f_i$.
Then (if player $i$ did not stop and return the output), the player waits for messages from a certain subset of the other 
players, this subset being also determined by the current view of the player. 
That is,  the player reads a single message from each of the incoming links that connect it  to that subset of other players;
If for a certain such link no message
is available, then the player waits until such message is available (i.e., arrives). 
Then the (local) round of player $i$ terminates.\footnote{The fact that the receiving of the incoming messages comes 
as the last step of  the (local) round comes only to emphasize that the sending of the messages and the output 
are a function of only the messages received in previous (local) rounds.} To make it possible for the player to identify 
the arrival of the {\em complete} message that it waits for, we require that each message sent by a player in the protocol 
is self-delimiting.
 
Denote by ${\cal D}_i^\ell$ the set of possible views of player $i$ at the end of local round 
$\ell$, $\ell \geq 0$, where the beginning of 
the protocol is considered round $0$. \\
Formally, a {\bf protocol} $\pi$ is defined by  a set of local programs, one for each player $i$,  where the local 
program of player $i$ is defined by a sequence of functions, parametrized by the index of the {\em local} 
round $\ell$, $\ell\geq 1$:
\LongVersion
\begin{itemize}
\LongVersionEnd
\ShortVersion
\begin{compactitem}
\ShortVersionEnd
\item $S_i^{\ell,s}: {\cal D}_i^{\ell-1} \rightarrow 2^{\{1,\ldots,k\}\setminus\{i\}}$, defining the set of players 
to which player $i$ \textit{sends} the messages.
\item $m_{i,j}^\ell:  {\cal D}_i^{\ell-1} \rightarrow \{0,1\}^{*}$, such that for any $D_i^{\ell-1} \in{\cal D}_i^{\ell-1}$, 
if  $j \in S_i^{\ell,s}(D_i^{\ell-1})$,  then $m_{i,j}^\ell(D_i^{\ell-1})$ is 
the content of the message player $i$ sends to player $j$. Each such message is self-delimiting.
\item $O_i^\ell: {\cal D}_i^{\ell-1} \rightarrow \{0,1\}^{*} \cup \{\bot\}$, defining whether or not  the local 
program of player $i$ stops and 
the player returns its output, and what is that output.
 If the value is $\bot$ then no output occurs. If the value is $y \in \{0,1\}^*$, then the 
local program stops and the player returns the value $y$. 
\item $S_i^{\ell,r}: {\cal D}_i^{\ell-1} \rightarrow 2^{\{1,\ldots,k\}\setminus\{i\}}$, 
defining the set of players from which player $i$ waits to \textit{receive} a message.
\LongVersion
\end{itemize}
\LongVersionEnd
\ShortVersion
\end{compactitem}
\ShortVersionEnd

To define the {\bf transcript} of a protocol we proceed as follows. 
We first define $k(k-1)$ basic transcripts $\Pi_{i,j}^r$, denoting the transcript of
the messages read by player $i$ from its link from player $j$, and another $k(k-1)$ basic transcripts 
$\Pi_{i,j}^{s}$, denoting the transcript of the messages sent by player $i$ on its link to player $j$.\\
We then define the transcript of player $i$, $\Pi_i$, as the $2(k-1)$-tuple of the $2(k-1)$ basic transcripts
$\Pi_{i,j}^r,\Pi_{i,j}^s$, $j  \in \Ni{1}{k}\setminus \{i\}$. 
The transcript of the whole protocol $\Pi$ is defined as the $k$-tuple of the $k$ player transcripts $\Pi_i$, $i \in \Ni{1}{k}$.
We denote by $\Pi_i(x,r)$ the transcript of player $i$ when protocol $\pi$ is run on input $x$ and on
randomness (public and private of all players) $r$. By $\Pi_i^{\ell}(x,r)$ we denote $\Pi_i(x,r)$ modified such 
that all the messages that player $i$ sends in local rounds $\ell'>\ell$, and all the messages that player $i$ reads
in local rounds $\ell'>\ell$ are eliminated from the transcript.
Observe that while $\Pi_{i,j}^r$ is always a prefix of $\Pi_{j,i}^{s}$, the definition
of a protocol does not imply that they are equal. Further observe that each bit sent in $\pi$ appears in $\Pi$ at most twice.
 
We note that while seemingly the model that we introduce here is the same as the one used in
~\cite{KRU}, there are important differences between the models, and that these differences  are crucial for the properties
 that we prove in the present paper to hold. See Section~\ref{sec:comparison_models} for a comparison.

For a $k$-party protocol $\pi$ we  denote the set of possible inputs as $\calX$, and 
denote the projection of this set on the $i$'th coordinate (i.e., the set of possible inputs for player $i$)
by $\calX_i$.  Thus $\calX \subseteq \calX_1\times\cdots\times\calX_k$.
The set of possible transcripts for a protocol is  denoted $\calT$, and the projection 
of this set on the $i$'th coordinate (i.e., the set of possible transcripts of player $i$) is 
denoted $\calT_i$. Observe that $\calT \subseteq \calT_1\times\cdots\times\calT_k$.

        Furthermore, in the course of the proofs, we sometimes consider a protocol that does {\em not} have access to
public randomness (but may have private randomness). We call such protocol a {\em private-coins protocol}.

We now formally define the notion of a protocol computing a given function with certain bounded error.
We will give most of the following definitions for the case where all functions $f_i$ are the same function, that we 
denote by $f$. The definitions in the case of family of functions are similar.

\begin{definition}
For a given $0 \leq \epsilon <1 $, a protocol $\pi$ $\epsilon$\textit{-computes} a function $f$ if 
for all $x \in \Pi_{i=1}^{k}\calX_i$:
\LongVersion
\begin{itemize}
\LongVersionEnd
\ShortVersion
\begin{compactitem}
\ShortVersionEnd
\item  For all possible assignments for the random sources $R_i$, $ 1\leq i\leq  k$, and $R^p$,  
every player eventually stops and returns an output.
\item With probability at least $1-\eps$  (over all random sources) the following event occurs:
each player $i$ outputs the value $f(x)$, i.e., the correct value of the function.
\LongVersion
\end{itemize}
\LongVersionEnd
\ShortVersion
\end{compactitem}
\ShortVersionEnd
\end{definition}

 The {\bf communication complexity}  of a protocol is defined as the worst case, 
over the possible inputs and the possible randomness, of the number
of bits sent by all players. 
For a protocol $\pi$ we denote its communication complexity by $\CC(\pi)$.
For a given function $f$ and a given $0 \leq \epsilon <1$, we denote by $\CC^{\epsilon}(f)$ the 
{\em $\epsilon$-error communication
complexity of $f$}, i.e., 
$\CC^{\epsilon}(f) = \inf\limits_{\pi~\text{$\epsilon$-computing}~f }\CC(\pi)$.

In addition to the notion of a protocol computing a function, we also consider the notion of \textit{external computation} by  a protocol.
\begin{definition}\begin{sloppypar}
For a given ${0 \leq \epsilon <1}$, a protocol $\pi$  {\em externally}  $\epsilon$-computes $f$ 
if there exists a deterministic  function $\theta$ taking as input the possible transcripts of $\pi$ and 
verifying ${\forall~x \in \calX, ~\Pr[\theta(\Pi(x)) = f(x)] \ge 1 - \epsilon}$.\end{sloppypar}
\end{definition}

We remark that an $\epsilon$-computing protocol can be converted into an  externally  $\epsilon$-computing protocol keeping the same communication complexity 
(up to a multiplicative factor of $2$, and an additive term being the size of the output): 
every bit sent in the original protocol is replaced by two identical bits; when some player, w.l.o.g. 
player $1$, outputs the value of the function, it concatenates to the string it has to send to some other player, w.l.o.g.  player $2$, a string composed 
of the bits $01$ followed  by the value of the function it outputs (the concatenated 
new string is to be sent as a self delimiting message as all other messages in the protocol).
The new protocol operates as the original one, and will be an
 externally computing protocol. Thus, for any function with constant-size output (in particular boolean functions),
a lower bound on the communication complexity of externally $\epsilon$-computing protocols implies the same (up to a constant factor of $2$) lower bound for 
$\epsilon$-computing protocols.

We further remark that an  externally  $\epsilon$-computing protocol is a weaker notion than an $\epsilon$-computing protocol, as it only requires that the  
 the function can be computed from the transcript, not that all, or any specific, players can compute the function. For example,  
for the function $f=x_1$, $x_i\in \{0,1\}^n$ (i.e., the value  of the function is the input  of player $1$)
a protocol where player $1$ sends to player $2$ its input $x_1$ is a $0$-externally computing protocol 
with communication complexity $n$. But, a $0$-computing protocol for $f$ 
requires communication complexity of $n(k-1)$ in order  that all players can compute the function.
 
\vspace{0.4cm}

Finally, we give a proposition that relates the communication complexity of a $k$-party protocol $\pi$ to the entropy of
the transcripts of the protocol $\pi$.
\begin{proposition}
\label{prop:entropy_and_cc}
Let the input to a $k$-party protocol $\pi$ be distributed according to an arbitrary distribution.
Then, 
$ \sum_{i=1}^k H(\Pi_i) \leq 4\cdot \CC(\pi) + 4 k^2$,
where the entropy is according to the input distribution  and the randomization of protocol $\pi$.
\end{proposition}
\begin{proof}
We first encode $\Pi_i$, for any $i$, into a variable $\Pi_i'$ such that the set of possible values of $\Pi_i'$ is a prefix-free set
of strings.
Observe that the transcript
$\Pi_i$ is composed of a number of basic transcripts:
 for every $j \in \Ni{1}{k} \setminus \{i\}$, a  pair of transcripts of messages, 
 $\Pi_{i,j}^s$, $\Pi_{i,j}^r$
containing the messages sent by player $i$ to player $j$, and the messages read by player $i$ from player $j$, respectively.
We convert $\Pi_i$ into $\Pi_i'$ as follows:  In each one of the above $2(k-1)$ components we
replace every bit $b\in \{0,1\}$  by $b.b$, and then add at the end of the component the two bits $01$. We then 
concatenate all components in order. Clearly this a one-to-one encoding, and the set of possible values of 
$\Pi_i'$ is a prefix-free set of strings.

\begin{sloppypar}Defining ${|\Pi_i| = \sum\limits_{j\neq i} |\Pi_{i,j}^s| + |\Pi_{i,j}^r|}$ and 
${|\Pi| = \sum\limits_{i=1}^k |\Pi_i|}$, we have 
${H(\Pi_i') = H(\Pi_i)}$, and ${\Exp[ | \Pi_i' |] = 2 \Exp[ | \Pi_i |] + 4(k-1)}$.
\end{sloppypar}
~~\\
We get
\begin{align*}
\sum\limits_{i=1}^k H(\Pi_i)&=  \sum\limits_{i=1}^k H(\Pi'_i)\\
&\le  \sum\limits_{i=1}^k \Exp[ | \Pi_i' |]  ~~~ \text{(by Theorem~\ref{thm:Shannon})}\\
&\le  \sum\limits_{i=1}^k \left( 2 \Exp[ | \Pi_i |] + 4(k-1)\right)\\
&\le 2 \cdot \Exp[ | \Pi |] + 4k^2 \\
&\le 4 \cdot \CC(\pi) + 4k^2~,
\end{align*}
where the last factor of $2$ is due to the fact that each message sent from, say, player $i$ to player $j$,
may appear in at most $2$ basic transcripts $\Pi_{i,j}^s$ and $\Pi_{j,i}^r$.
\end{proof}

\subsection{Comparison to other models}
\label{sec:comparison_models}
The somewhat restricted model (compared to the general asynchronous model) that we work with allows us to 
use information theoretic tools for the study of protocols in this model, and in particular to give lower bounds on
 the multi-party 
communication complexity. 
Notice that the general asynchronous model is problematic in this respect since one bit of communication
 can bring $\log k$ bits of information, because not only the content of the message, but also the identity of the
  sender may reveal information. 
Thus, information cannot be used as a lower bound on communication. In our case, the sets $S_i^{l,r}$ 
and $S_i^{l,s}$ are determined by the current view of the player, $\Pi$ contains only the content of the
 messages, and thus the desirable relation between the communication and the information is maintained. On 
 the other hand, our restriction is natural, does not seem to be very restrictive (practically all protocols in the 
 literature adhere to our
model), and does not exclude the existence of private protocols.
To exemplify why the general asynchronous  model is problematic  see
 Example~\ref{ex:problematic_protocol}.
 
  While the model that we introduce in the preset paper bears some similarities to the model used in~\cite{KRU},
   there are a number of important differences between them. First, the definition of the transcript is different, resulting in a 
   different  relation 
between the entropy of the transcript and the communication complexity. More important is the natural property of the model in 
the present paper that the local 
   program of a protocol in a given node ends its
    execution when it locally gives its output. It turns out that the very basic rectangularity property of protocols,
    used in many papers, holds in this case (and when the transcript is defied as we define in the present paper),
     while if the local protocol may continue to operate after output, there are 
    examples where this property does not hold. Thus, we view the introduction of the present model also as
     a contribution towards identifying the necessary features of a peer-to-peer model so that basic and useful
      properties of protocols hold in the peer-to-peer setting.

There has been a long series of works about multi-party communication protocols in different variants of 
models, for example  \cite{DF,CKS,Gro,Jay,PVZ,CRR,CR} (see~\cite{EOPV} for a comparison of a few of 
these models).
In the \textit{coordinator model} (cf.~\cite{DF,PVZ,BEOPV}), an additional player (the coordinator) with no input can
 communicate privately with each player, and the players can only communicate with the coordinator. 
We first note that the coordinator model does not yield exact bounds for the multi-party communication 
complexity in the peer-to-peer setting (neither in our model nor in the most general one). 
Namely, any protocol in the peer-to-peer model can be transformed into a protocol in the coordinator model with 
an $O(\log k)$ multiplicative factor in the communication complexity, by sending each message to the 
coordinator with an $O(\log k)$-bit label indicating its destination. This factor is sometimes necessary, 
e.g.,  for the {\tt permutation} functional
defined as follows: Given a permutation $\sigma:\Ni{1}{k}\rightarrow \Ni{1}{k}$, each player $i$ has as input a bit $b_i$ 
and $\sigma^{-1}(\sigma(i)-1)$ and $\sigma^{-1}(\sigma(i)+1)$  (i.e., each player has as input the indexes of the players
before and after itself in the permutation).\footnote{All additions are modulo $k$. This is a promise problem.} 
For player $i$ the function $f_i$ is defined as $f_i=b_{\sigma^{-1}(\sigma(i)+1)}$ 
(i.e., the value of the input bit of the next player in the permutation $\sigma$). Clearly in our model the 
communication complexity of this function is $k$ 
(each player sends its
input bit to the correct player), and the natural protocol is valid in our model. On the other hand, in the coordinator model
$\Omega(k \log k)$ bits of communication are necessary.
But this multiplicative factor between the complexities in the two models is not always necessary: the communication 
complexity of the parity function $\XOR$ is $\Theta(k)$ both in the peer-to-peer model and in the coordinator model. 

Moreover, when studying private protocols in the multi-party setting, the coordinator model does not offer any 
insight. In the coordinator model, described in \cite{DF} and used for instance in \cite{BEOPV}, 
if one does not impose any  
 privacy requirement with respect to the coordinator, it is trivial to have a private protocol by all players 
sending their input to the coordinator, and the coordinator returning the results to the players. If there is a privacy
 requirement with respect to the coordinator, then if there is a random source shared by all the players (but not
  the coordinator), privacy is always possible using the protocol of~\cite{FKN}. If no such source exists, privacy  is 
 impossible in general.
This follows from the results of Braverman et al.~\cite{BEOPV} who show a non-zero lower bound on the total
 internal information complexity of all parties (including the coordinator) for the function {\em Disjointness} in that model. 
Our model, on the other hand, does allow for the private computation of any function~\cite{BGW,CCD,AsharovL17}.

It is worthwhile to contrast our model, and the communication complexity measure that we are concerned with,
with work in the so-call congested-clique model that has gained increasing attention in the distributed computation
literature (cf.~\cite{KorhonenS17,KorhonenS17_arxiv}). While both models are based on a communication 
network in the form 
of a complete graph
 (i.e., every player can send messages to any other player, and these messages can be different) there are 
 two significant differences between them. Most of the works
 in the congested clique model deal with graph-theoretic problems and the input to each player 
  is related to
 the adjacency list of a node (identified with that player) in the input graph, while in  our model the 
 input is not associated in any way with the communication graph.
 More importantly,
 the 
 congested clique model
is a {\em synchronous} model while ours is an asynchronous one. 
This brings about a major difference between the complexity measures studied in each  of the models. Work in
 the congested clique model is concerned with giving bounds on the number of rounds necessary to 
  fulfill a certain task under the condition that in each round each player can send to any other player a 
  limited number of bits (usually $O(\log k)$ bits).  The measure of communication complexity, that 
  is of interest to us in the 
  present paper, deals with the total number of communication bits necessary to
fulfill a certain task in an asynchronous setting without any notion of global rounds.\footnote{Any function can be 
computed in the congested clique model with $O(k)$ 
 communication complexity (at a cost of having many rounds) by each player, having input $x$, sending  
 a single bit to player $1$ only at round number $x$. On the other hand, in the asynchronous model any function
 can be computed in a single ``round'' (at a cost of high communication complexity) by each player sending 
 its whole input to player $1$.}

\section{Tools for the study of multi-party communication protocols}\label{sec:Prelim}

In this section we consider two important tools for the study of peer-to-peer multi-party communication protocols.
First, we define and introduce an  information theoretic measure that we call 
{\em Multi-party Information Cost} (MIC); we later use it to prove our lower bounds.
Then, we prove, in the peer-to-peer multi-party model that we define,  the so-called rectangularity property of
communication protocols, that we also use in our proofs.

\subsection{Multi-party Information Cost}
\label{subsec:mic}

We now introduce  an information theoretic measure for multi-party peer-to-peer protocols that we later show
to be  useful  for proving lower bounds on the communication complexity of multi-party peer-to-peer protocols.
We note that a somewhat similar measure
was proposed in~\cite{BEOPV} for the coordinator model, but, to the best of our knowledge, never found an application as a 
tool in a proof of a lower bound.

\begin{definition}
For any $k$-player protocol $\pi$ and any input distribution $\mu$, 
we define the \textit{multi-party information cost} of $\pi$:
$$\MIC_{\mu}(\pi)=\sum\limits_{i=1}^k  \left(I(X_{-i} ; \Pi_i \mid X_i R_i) + I(X_i ; \Pi_i \mid X_{-i} R_{-i})\right)~.$$
\end{definition}

Observe that 
the second part of each of the $k$ summands can be interpreted as  the information that player $i$ ``leaks'' to the other
players on its input.
While the ``usual'' intuitive  interpretation of  two-party $\IC$ is ``what Alice learns on Bob's input plus what Bob
learns on Alice's input'', one can also interpret two-party $\IC$ as ``what Alice learns on Bob's input plus what Alice leaks
on her input''. 
Thus, $\MIC$ can be interpreted as summing over all players $i$ of ``what  player $i$ learns on 
the other players' inputs, plus what player $i$ leaks on its input.'' Indeed, the expression defining $\MIC$ is equal 
to  the sum, over all players $i$, of the two-party $\IC$ for the two-party protocol that results from collapsing all 
players, except $i$, into one virtual player. Thus, for number of players $k=2$, $\MIC=2\cdot \IC$.
We note that  defining our measure without the private randomness in the condition of the mutual 
information expressions would yield the exact same measure (as is the case for $2$-party $\IC$); we prefer however
to define $\MIC$ with the randomness in the conditions, as we believe that it allows one to give shorter, but still 
clear and accurate, proofs. 

On the other hand observe that the second of the two mutual information expressions has
$X_{-i}$ in the condition, contrary to a seemingly similar measure used in~\cite{BEOPV} (Definition 3 
in~\cite{BEOPV}).
 Our measure
 is thus ``internal'' 
in nature, while the one of~\cite{BEOPV} has an ``external'' component.

Further observe that the summation, over all players, of each one of the two  mutual information expressions
alone would not yield a  measure useful for proving lower
bounds on the communication complexity of functions.
The first  mutual information expression would yield a measure  for functions that would never be higher than the 
entropy of the function at hand, due to the existence
of private protocols for all functions~\cite{BGW,CCD}. For the second mutual information expression there are functions for which that measure
 would be far too low compared to the communication complexity: e.g., the function $f=x_1$, $x\in \{0,1\}^n$ (i.e., the value 
of the function is the input  of player $1$); in that case the measure would equal only $n$, while the communication
 complexity of that function is $\Omega(kn)$.

We now define the multi-party information complexity of a function. 
\begin{definition}
For any function $f$, any input distribution $\mu$, and any $0 \leq \epsilon \leq 1$, we define the quantity
$$\MIC^\epsilon_\mu(f) =  \inf\limits_{\pi\text{ $\epsilon$-computing }f} \MIC_\mu(\pi)~.$$
\end{definition}

\begin{definition}
For any $f$, and any $0 \leq \epsilon \leq 1$, we define the quantity
$$\MIC^\epsilon(f) =  \inf\limits_{\pi\text{ $\epsilon$-computing }f} \sup\limits_{\mu}~\MIC_\mu(\pi)~.$$
\end{definition}

We now claim that the multi-party information cost  and the communication complexity of a protocol
are related, as formalized by the following lemma.
\begin{lemma}\label{lem:CC>=MIC}\begin{sloppypar}
For any $k$-player protocol $\pi$, and for any input distribution $\mu$,
$$\CC(\pi) \ge \frac{1}{8}\MIC_{\mu}(\pi) - k^2~.$$
\end{sloppypar}\end{lemma}

\begin{proof}
\begin{align*}
\MIC_{\mu}(\pi)&=\sum\limits_{i=1}^k  \left(I(X_{-i} ; \Pi_i \mid X_i R_i) + I(X_i ; \Pi_i \mid X_{-i} R_{-i})\right)\\
&\le 2 \sum\limits_{i=1}^k H(\Pi_i) \\
&\le 8 \cdot \CC(\pi) + 8k^2~,
\end{align*}
where the first inequality follows from Proposition~\ref{prop:Hcondi}, and the last one 
from Proposition~\ref{prop:entropy_and_cc}.
\end{proof}

We now show that the multi-party information cost satisfies a direct sum property for product distributions. In what 
 follows, the notation $f^{\otimes n}$ denotes the task of computing $n$ instances of $f$, 
 where the  requirement from an  
$\epsilon$-computing protocol is that each instance is computed correctly with probability at least $1-\epsilon$ 
(as opposed to  the stronger requirement  that the whole vector of instances is computed correctly 
with probability at least  $1-\epsilon$).

\begin{theorem}\label{thm:dsmic}
For any protocol $\pi$  $\epsilon$-computing a function $f^{\otimes n}$,
and for any product distribution $\mu$ for the 
input, 
there exists a protocol $\pi'$  $\epsilon$-computing $f$
such that,
$$\MIC_{\mu^n}(\pi) \ge n \cdot \MIC_{\mu}(\pi')~.$$
\end{theorem}

\begin{proof}
We define $\pi'$ on input $(Y_i)_{i\in\Ni{1}{k}}$ as follows. We denote by $R'^{p}$ the public randomness available to
 the players, and by $R'_{i}$ the private randomness available to the players.

We consider the public and private randomness $R'_i$, $1\leq i \leq k$, and $R'^p$ as strings of random bits.
The players first use  the first bits of the public randomness  to publicly sample a random index $L$ uniformly
 in $\Ni{1}{n}$, and 
define $X_i^L=Y_i$.
The players then,  using the next random bits of the public randomness, publicly sample,  for every $d<L$, 
$X^d$ according to $\mu$. 
Each player $i$ then, using the first bits of its private randomness,  samples privately, for every $d>L$, 
$X_i^d$ according to $\mu$.
The players then run $\pi$ on input $X$.  They output, as the output of $\pi'$, the $L$'th coordinate of the output of $\pi$.
Observe that $\pi'$ has error at most $\epsilon$, and that if the input to $\pi'$ is distributed according to $\mu$, then 
the input of
 $\pi$ is distributed according to $\mu^n$.

Note that there is no extra communication in $\pi'$ compared to $\pi$,
 only some (private and public) sampling. Therefore we have 
 $\Pi'_i=\Pi_i$ for every $1 \leq i \leq k$.  We further denote by $R^p$ the random bits of $R'^{p}$
 beyond those used by the public sampling at the start of $\pi'$. Similarly, we denote by $R_{i}$, $1 \leq i \leq k$,
 the  random bits of $R'_{i}$ beyond those used by the private sampling at the start of $\pi'$.

We now show that $\MIC_\mu(\pi') = \frac{1}{n}\MIC_{\mu^n}(\pi)$. In what follows we explicitly state 
the public randomness next to the transcript. Thus, 	
$$\MIC_{\mu}(\pi')=\sum\limits_{i=1}^k  \left(I(Y_{-i} ; R'^p\Pi_i' \mid Y_i R'_i ) + I(Y_i ; R'^p\Pi_i' \mid Y_{-i} R'_{-i} )\right).$$

We have, for every player $i$,
\begin{align*}
I(Y_{-i} ; R'^p\Pi_i' \mid Y_i R'_i ) &=  I(Y_{-i} ;  LX^{<L}R^p \Pi_i \mid  Y_i X_i^{>L} R_i ) \text{~~~(making explicit the sampling 
from $R'_i$,$R'^p$)}\\
&= I(Y_{-i} ;  X^{<L}R^p \Pi_i \mid  Y_i X_i^{>L} R_i ) \text{~~~(because $I( Y_{-i} ; L \mid Y_i X_i^{>L} R_i X^{<L}R^p\Pi_i)=0$)}\\
&=  I(Y_{-i} ; R^p \Pi_i \mid  Y_iX_i^{>L} R_i X^{<L} ) \text{~~~(because $Y_{-i}$ and $X^{<L}$ are independent)} \\
&= \Exp\limits_\ell[I(X^\ell_{-i} ; R^p\Pi_i \mid X^\ell_i X_i^{>\ell} R_i X^{<\ell} )]   \\
&=\Exp\limits_\ell[I(X^\ell_{-i} ; R^p\Pi_i \mid X^\ell_i X_i^{>\ell} R_i X_i^{<\ell} X_{-i}^{<\ell} )]\\
&=\Exp\limits_\ell[I(X^\ell_{-i} ; R^p\Pi_i \mid X_i R_i  X_{-i}^{<\ell})]\\
&=\frac{1}{n}\sum\limits_\ell[I(X^\ell_{-i} ; R^p \Pi_i \mid X_i R_i X_{-i}^{<\ell})]\\
&=\frac{1}{n}I(X_{-i} ; R^p \Pi_i \mid X_i R_i )\text{~~~(chain rule)~,}
\end{align*}
and
\begin{align*}
I(Y_i ;  R'^p\Pi_i' \mid Y_{-i} R'_{-i})&=  I(Y_i ; LX^{<L}R^p\Pi_i \mid  Y_{-i} X_{-i}^{>L}R_{-i}) \text{~~~(making explicit the sampling 
from $R'_i$,$R'^p$)}\\
&=I(Y_i ; X^{<L}R^p\Pi_i \mid  Y_{-i} X_{-i}^{>L}R_{-i})  \text{~~~(because $I(Y_i ; L \mid  Y_{-i} X_{-i}^{>L}R_{-i}X^{<L}R^p\Pi_i )=0$)} \\
&= I(Y_i ; R^p\Pi_i \mid  Y_{-i} X_{-i}^{>L}R_{-i}X^{<L} )\text{~~~(because $Y_{i}$ and $X^{<L}$ are independent )} \\
&=\Exp\limits_\ell[I(X^\ell_i ; R^p\Pi_i \mid X^\ell_{-i} X_{-i}^{>\ell} R_{-i} X_{-i}^{<\ell} X_{i}^{<\ell} )]\\
&=\Exp\limits_\ell[I(X^\ell_i ; R^p\Pi_i \mid X_{-i} R_{-i}  X_{i}^{<\ell})]\\
&=\frac{1}{n}\sum\limits_\ell[I(X^\ell_i ; R^p\Pi_i \mid X_{-i} R_{-i}  X_{i}^{<\ell})]\\
&=\frac{1}{n}I(X_i ; R^p\Pi_i \mid X_{-i} R_{-i} )\text{~~~(chain rule).}
\end{align*}

Summing over $i\in\Ni{1}{k}$ concludes the proof.
\end{proof}

\subsection{The rectangularity property}
\label{sec:multi-party_protocols}

\LongVersion
\paragraphA{Rectangularity.}
\LongVersionEnd
The \textit{rectangularity property} (or \textit{Markov property}) is one of the key properties that follow 
from the structure 
and definition of (some) protocols. 
For randomized protocols it was introduced in the two-party setting and in the multi-party
 blackboard model in \cite{BJKS}, and in the coordinator model in \cite{BEOPV}. We prove a similar rectangularity 
 property in the peer-to-peer model that we consider in the present paper. 
 
 We note that the proof of this 
 property in the peer-to-peer model makes explicit use of the specific
 properties of the model we defined:
  the proof that follows explicitly uses the definition of the transcript on an edge by edge basis as in our model, 
 as well as the fact that a player returns and stops as one operation. One can build examples  where if
 any of these two properties does not hold, then the 
 rectangularity property of protools does not hold.
 Thus we view the following proof of rectangularity in our model also as an identification of model properties  needed 
 for  the useful rectangularity property of multiparty peer-to-peer protocols to hold.
 
To define this property,
for any transcript $\overline{\tau}\in\calT_i$, let $\calA_i(\overline{\tau})=\{(x,r) \mid \Pi_i(x,r)=\overline{\tau}\}$ 
(i.e., the set of input, randomness pairs that lead to transcript $\overline{\tau}$),
and define the projection of $\calA_i(\overline{\tau})$ on coordinate $i$ as 
\LongVersion
$$\calI_i(\overline{\tau})=\{(x',r'), \exists~(x,r)\in\calA_i(\overline{\tau}), x'=x_i ~\&~ r'=r_i\}~,$$
\LongVersionEnd
\ShortVersion
$\calI_i(\overline{\tau})=\{(x',r'), \exists~(x,r)\in\calA_i(\overline{\tau}), x'=x_i ~\&~ r'=r_i\}$,
\ShortVersionEnd
and the projection of $\calA_i(\overline{\tau})$ on the complement of coordinate $i$ as 
\LongVersion
$$\calJ_i(\overline{\tau})=\{(x',r'), \exists~(x,r)\in\calA_i(\overline{\tau}), x'=x_{-i} ~\&~ r'=r_{-i}\}~.$$
\LongVersionEnd
\ShortVersion
$\calJ_i(\overline{\tau})=\{(x',r'), \exists~(x,r)\in\calA_i(\overline{\tau}), x'=x_{-i} ~\&~ r'=r_{-i}\}$.
\ShortVersionEnd
Similarly, for any transcript $\tau\in\calT$, let 
$\calB(\tau)=\{(x,r) \mid \Pi(x,r)=\tau)\}$, and for any player $i$, let
$\calH_i(\tau)=\{(x',r'),\exists~(x,r)\in\calB(\tau),x'=x_{-i}~\&~r'=r_{-i}\}$.

We start by proving a combinatorial property of transcripts of communication protocols, 
which intuitively follows from the fact that each player has access to only its own input and private randomness. 
The proof of this property is technically more involved  compared to the analogous property
 in other settings, since the structure of protocols and the manifestation of the transcripts
  in the peer-to-peer setting are more flexible than in the other settings.

\begin{lemma}
\label{le:det_rectangularity}
Let $\pi$ be a $k$-player private-coins protocol
\LongVersion
\footnote{Recall that a private-coins protocol does not have 
access to public randomness, but may have private randomness.}  
\LongVersionEnd
with inputs from 
$\calX= \calX_1\times\cdots\times\calX_k$. Let $\calT$ denote the set of possible transcripts 
of $\pi$, and for $i\in\Ni{1}{k}$ let $\calT_i$ denote the set of possible transcript observed by player 
$i$, so that $\calT\subseteq \calT_1\times\cdots\times\calT_k$.
Then, $\forall~ i\in \Ni{1}{k}$:
\LongVersion
\begin{itemize}
\LongVersionEnd
\ShortVersion
\begin{compactitem}
\ShortVersionEnd
\item $\forall~\overline{\tau} \in\calT_i,~~\calA_i(\overline{\tau}) = \calI_i(\overline{\tau}) \times \calJ_i(\overline{\tau})$. 
\item $\forall~\tau\in\calT, ~~\calB(\tau)=\calI_i(\tau_i)\times\calH_i(\tau)$.
\LongVersion
\end{itemize}
\LongVersionEnd
\ShortVersion
\end{compactitem}
\ShortVersionEnd
\end{lemma}
\begin{proof}
We start by proving the first claim.
Since the other inclusion is immediate from the definition, we only need to show that 
$$\forall~\overline{\tau} \in \calT_i, ~\calI_i(\overline{\tau}) \times \calJ_i(\overline{\tau}) \subseteq \calA_i(\overline{\tau})~.$$ 
\begin{sloppypar}To this end 
take an arbitrary $(x_i,r_i) \in \calI_i(\overline{\tau})$ and an arbitrary $(x_{-i},r_{-i}) \in \calJ_i(\overline{\tau})$.
Since ${(x_i,r_i)\in\calI_i(\overline{\tau})}$,  we have that 
$\exists~(\tilde{x},\tilde{r})\in \calA_i(\overline{\tau})$ s.t. $ x_i=\tilde{x}_i ~\&~ r_i=\tilde{r}_i$.
Similarly, since $(x_{-i},r_{-i}) \in \calJ_i(\overline{\tau})$, \end{sloppypar}
$$\exists~(\hat{x},\hat{r})\in\calA_i(\overline{\tau}) \mid x_{-i}=\hat{x}_{-i} ~\&~ r_{-i}=\hat{r}_{-i}~.$$
Let 
 $(x,r)$ be $((x_i,r_i),(x_{-i},r_{-i}))\in\calI_i(\overline{\tau}) \times \calJ_i(\overline{\tau})$.
We will now show that $$(x,r) \in \calA_i(\overline{\tau})~.$$ 

Let $L$ be the number of local rounds of player $i$ in the run of $\pi$ on input $(x,r)$.
We will show by induction on the index  of the local round of player $i$ that for any $\ell \leq L$, 
$\Pi_i^{\ell}(x,r)= \Pi_i^{\ell}(\tilde{x},\tilde{r})$. 
 Observe that whether or not the protocol of a player stops and returns its output  at a given round is a function 
 of its input and
  its transcript until that round, as well as its private randomness.  Therefore, 
  since  the protocol of player $i$ stops and returns its value 
at local round $L$  if the input is $(x,r)$, it will follow from $\Pi_i^{L}(x,r)= \Pi_i^{L}(\tilde{x},\tilde{r})$
that player $i$ stops and returns its output at local round $L$ also when the input is $(\tilde{x},\tilde{r})$. We will thus get that 
$\Pi_i(x,r)=\overline{\tau}$, and hence $(x,r) \in \calA_i(\overline{\tau})$.
 
The base of the induction, for $\ell=0$, follows since the transcript is empty. We now prove the claim for $\ell+1 \leq L$, based on the induction hypothesis that the claim holds for $\ell$.\footnote{Note that $\Pi_i(x,r)$ by itself does not define which messages are sent/read in which local round.}

The messages that player $i$ sends at local round $\ell+1$ are a function of $x_i$, $r_i$ and $\Pi_i^{\ell}(x,r)$.
As $x_i=\tilde{x}_i$ and $\tilde{r}_i=r_i$, and using the induction hypothesis, we get that the messages sent by player 
$i$ at local round $\ell+1$ are the same in $\pi_i(x,r)$ and in $\pi_i(\tilde{x},\tilde{r})$.
  
For the same reason we also get that the set of players from which player $i$ waits for a message at round $\ell+1$ is the same when $\pi$ is run on input  in $(x,r)$ and on input $(\tilde{x},\tilde{r})$.
  
We now claim that the messages read by player $i$ at round $\ell+1$ are the same when $\pi$ is run on input $(x,r)$ and on input $(\tilde{x},\tilde{r})$. To this end we define an imaginary ``protocol'' $\psi$ where player $i$ sends in its first local round all the messages that it sends  in $\overline{\tau}$, and the players in $Q_i=\Ni{1}{k}\setminus\{i\}$ run $\pi$.\footnote{Technically speaking, this is not a protocol according to our definition as more than one message may be sent in a single round on a single link.}
Player $i$ sends the messages on each link according to the order in $\overline{\tau}$.\footnote{Recall that a transcript of a players is a $2(k-1)$-tuple of transcripts, one for each of its $2(k-1)$ directed links.}
The messages that the players in $Q_i$ send in each of their local rounds are a function of their inputs (and their local randomness)
and the messages  they read from the links that connect to player $i$. Since $\Pi(\hat{x},\hat{r})=\overline{\tau}$,
we can conclude that in $\psi$ (when the input is $(\hat{x},\hat{r})$) the messages sent by the players in $Q_i$ (in particular, to player $i$) are the same as those sent in $\pi$ on input $(\hat{x},\hat{r})$.

Recall that we have proved above that when $\pi$ is run on $(x,r)$, the messages player $i$ sends up to round $\ell+1$ are consistent with $\overline{\tau}$. We therefore can consider now a ``protocol'' $\psi'$ which is the same as $\psi$ with the only difference that player $i$ sends (in its first local round) only the messages of $\overline{\tau}$ it would have sent in $\pi(x,r)$ until (and including) round $\ell+1$ (and not all the message it sends in $\overline{\tau}$). It follows that in $\psi'$, when run on input $(\hat{x},\hat{r})$, the sequences of messages sent from the players in $Q_i$ to $i$ are a prefix of the sequences they send in $\psi$.
Since $x_{-i}=\hat{x}$ and  $r_{-i}=\hat{r}$, the same claim holds when $\psi'$ is run on $(x,r)$. Observe now that when $\pi$ is run on $(x,r)$, at the time where player $i$ is waiting at local round $\ell+1$ for incoming messages it, has sent exactly the messages that player $i$ sends in $\psi'$.

Using the induction hypothesis $\Pi_i^{\ell}(x,r)= \Pi_i^{\ell}(\tilde{x},\tilde{r})$, 
the fact hat $x_i=\tilde{x}_i$ and $r_i=\tilde{r}_i$, and
the fact that the set of players from which player $i$ waits for a message at local round $\ell+1$ is the same for input 
$(x,r)$ and $(\tilde{x},\tilde{r})$,
we can conclude that the messages that player $i$ reads while waiting for messages at local round $\ell+1$ when 
$\pi$ is run on $(x,r)$ are consistent with the messages it would read when $\pi$ is run on $(\tilde{x},\tilde{r})$. Since 
player $i$ running $\pi$ must, by the definition of a protocol, reach its ``return'' statement, it must receive messages from 
all the players it is waiting for. We therefore conclude that the messages read by player $i$ in local round 
$\ell+1$ when $\pi$ is run on $(x,r)$ are the same as those it read when run on $(\tilde{x},\tilde{r})$.

Together with the induction hypothesis, and the fact (proved above) that the messages sent by player $i$ at local round $\ell+1$ are the same when $\pi$ is run on in $(x,r)$ and on $(\tilde{x},\tilde{r})$, we have that 
$\Pi_i^{\ell+1}(x,r)= \Pi_i^{\ell+1}(\tilde{x},\tilde{r})$.

We now prove the second claim. We only need to show that
$$\forall~\tau\in\calT,~ \calI_i(\tau_i)\times\calH_i(\tau)\subseteq\calB(\tau)~,$$
the other inclusion being immediate from the definitions, since $\calB(\tau) \subseteq \calA_i(\tau_i)$.

Take an arbitrary $(x_i,r_i) \in \calI_i(\tau_i)$ and an arbitrary $(x_{-i},r_{-i}) \in \calH_i({\tau})$.
Since $(x_{-i},r_{-i}) \in \calH_i({\tau})$, 
$\exists~(\hat{x},\hat{r})$ s.t. $\pi(\hat{x},\hat{r})=\tau$ and $x_{-i}=\hat{x}_{-i}$ and $ r_{-i}=\hat{r}_{-i}$.
Let  $(x,r)=((x_i,r_i),(x_{-i},r_{-i}))$. Since $\calB(\tau)\subseteq\calA(\tau_i)$, 
we have $\calH_i(\tau)\subseteq\calJ_i(\tau_i)$. 
Thus, using the first claim,
$$\calI_i(\tau_i)\times\calH_i(\tau)\subseteq\calI_i(\tau_i)\times\calJ_i(\tau_i)\subseteq\calA_i(\tau_i)~,$$ 
and $\Pi_i(x,r)=\tau_i$.  It remains to show that  $\forall~j\neq i, ~\Pi_j(x,r)=\tau_j$. 

Consider the two runs of protocol $\pi$ on the input $(x,r)$ and on the input $(\hat{x},\hat{r})$. We have that 
$\Pi(\hat{x},\hat{r})=\tau$, and that $\Pi_i(x,r)=\tau_i$. Since $~x_{-i}=\hat{x}_{-i}$ and $ r_{-i}=\hat{r}_{-i}$,
we have that also for all $j \neq i$ $\Pi_j({x},{r})=\Pi_j(\hat{x},\hat{r})=\tau_j$. It follows that 
$(x,r) \in \calB(\tau)$ as needed.

\end{proof}

We now prove the \textit{rectangularity property of randomized protocols} in the peer-to-peer setting.
\LongVersion
It follows from Lemma~\ref{le:det_rectangularity} and straightforward calculations. 
The full proof is given in the appendix.
\LongVersionEnd

\begin{lemma}\label{lem:rectang}
Let $\pi$ be a $k$-player private-coins protocol with inputs from 
$\calX= \calX_1\times\cdots\times\calX_k$. Let $\calT$ denote the set of possible transcripts 
of $\pi$, and for $i\in\Ni{1}{k}$ let $\calT_i$ denote the set of possible transcript observed by player 
$i$, so that $\calT\subseteq \calT_1\times\cdots\times\calT_k$.
Then for every $i\in\Ni{1}{k}$, there exist functions
$q_i:\calX_i\times\calT_i\rightarrow[0,1]$, $q_{-i}:\calX_{-i}\times\calT_i\rightarrow[0,1]$ and 
$p_{-i}:\calX_{-i}\times\calT\rightarrow[0,1]$ such that 
\LongVersion
$$\forall~x\in\calX, \forall~\tau=(\tau_1,\dots,\tau_k)\in\calT, \Pr[\Pi_i(x)=\tau_i]=q_i(x_i,\tau_i)q_{-i}(x_{-i},\tau_i)~,$$
\LongVersionEnd
\ShortVersion
$\forall~x\in\calX, \forall~\tau=(\tau_1,\dots,\tau_k)\in\calT, \Pr[\Pi_i(x)=\tau_i]=q_i(x_i,\tau_i)q_{-i}(x_{-i},\tau_i)$,
\ShortVersionEnd
and
\LongVersion
$$\forall~x\in\calX, \forall~\tau=(\tau_1,\dots,\tau_k)\in\calT, \Pr[\Pi(x)=\tau]=q_i(x_i,\tau_i)p_{-i}(x_{-i},\tau)~.$$
\LongVersionEnd
\ShortVersion
$\forall~x\in\calX, \forall~\tau=(\tau_1,\dots,\tau_k)\in\calT, \Pr[\Pi(x)=\tau]=q_i(x_i,\tau_i)p_{-i}(x_{-i},\tau)$.
\ShortVersionEnd
\end{lemma}

\ThmOnlyInLong
\LongVersion
The following lemma formalizes the fact that the distribution of the transcript of a protocol that externally-computes 
a function $f$ must differ on two inputs with different values of $f$ (see also~\cite{BJKS}).
The proof is deferred to the appendix.
\begin{lemma}\label{lem:HelErr}
Let $f$ be a $k$-party function, and let $\pi$ be a protocol  externally $\epsilon$-computing $f$. 
If $x$ and $y$ are two inputs such that $f(x)\neq f(y)$, then $h(\Pi(x),\Pi(y)) \ge \frac{1-2\epsilon}{\sqrt{2}}$.
\end{lemma}
\LongVersionEnd

\ThmOnlyInLong
\LongVersion
\paragraphA{The Diagonal Lemma.}
The following lemma is often called the \textit{diagonal lemma}. It was proved in \cite{BJKS} for the two-party setting
 under the name of  the \textit{Pythagorean lemma}, 
 and in \cite{BEOPV} for the coordinator model. We show here that is also holds in the peer-to-peer model. 
This lemma 
follows from Lemma~\ref{lem:rectang} and  
Proposition~\ref{prop:Hellsquare} in the same way that its two-party analogue follows from the analogous lemma and proposition.  
For completeness we give the proof in the appendix.
For $x\in\{0,1\}^k$ and $b\in\{0,1\}$, let $x_{[i\leftarrow b]}$ represent the input obtained from $x$ 
by replacing the $i^\text{th}$ bit of $x$ by $b$.
\begin{lemma}\label{lem:diag}
Let $\pi$ be a $k$-party  private-coins protocol  taking input in $\{0,1\}^k$.
Then
${\forall~x\in\{0,1\}^k}$, $\forall~y\in\{0,1\}^k$, $\forall~i\in\Ni{1}{k}$,
$h^2(\Pi(x),\Pi(y)) \ge \frac{1}{2}\left[h^2(\Pi(x),\Pi(y_{[i\leftarrow x_i]}))+h^2(\Pi(x_{[i\leftarrow y_i]}),\Pi(y))\right]$.
\end{lemma}

\LongVersionEnd

\section{The function parity}
\label{sec:parity}

We now prove a lower bound on the multi-party peer-to-peer randomized communication complexity of the
 $k$-party $n$-bit parity function  $\XOR_k^n$, defined as follows:
each player $i$ receives $n$ bits $(x_i^p)_{p \in \Ni{1}{n}}$ and player $1$ has to output the bitwise sum modulo $2$ of the inputs, i.e.,
\LongVersion
$$\XOR_k^n(x)=\left( \oplus_{i=1}^k x_i^1, \oplus_{i=1}^k x_i^2,\ldots, \oplus_{i=1}^k x_i^n \right)$$
\LongVersionEnd
\ShortVersion
$\XOR_k^n(x)=\left( \oplus_{i=1}^k x_i^1, \oplus_{i=1}^k x_i^2,\ldots, \oplus_{i=1}^k x_i^n \right)$
\ShortVersionEnd
(the case where all $k$ players compute the   function is trivial).
To start, we prove a lower bound on the multi-party information complexity of the parity function, where each player has
 a single input bit. For simplicity we denote this function $\XOR_k$, rather than $ \XOR_k^1$.

\begin{theorem}\label{thm:micxor}
Let $\mu$ be the uniform distribution on $\{0,1\}^k$. 
Given any  fixed $0 \leq  \epsilon < \frac{1}{2}$,  for any protocol $\pi$  $\epsilon$-computing $\XOR_k$, it holds that
$\MIC_{\mu}(\pi) = \Omega(k)$.
\end{theorem}

\begin{proof}
\begin{align*}
\MIC_{\mu}(\pi) &= \sum\limits_{i=1}^k  \left(I(X_{-i} ; \Pi_i \mid X_i R_i) + I(X_i ; \Pi_i \mid X_{-i} R_{-i})\right)\\
&\ge \sum\limits_{i=2}^k I(X_i ; \Pi_i \mid X_{-i} R_{-i})\\
&= \sum\limits_{i=2}^k (I(X_i ; \Pi_i \mid X_{-i} R_{-i}) + I(X_i ; \Pi_1 \mid X_{-i} R_{-i} \Pi_i))
\text{~~~(as $H(\Pi_1 \mid X_{-i} R_{-i} \Pi_i) = 0$)}\\
&= \sum\limits_{i=2}^k I(X_i ; \Pi_1 \Pi_i \mid X_{-i} R_{-i}) \text{~~~(chain rule)}\\
&\ge \sum\limits_{i=2}^k I(X_i ; \Pi_1 \mid X_{-i} R_{-i})\\
&= \sum\limits_{i=2}^k (H(X_i \mid X_{-i} R_{-i}) - H(X_i \mid X_{-i} R_{-i} \Pi_1))\\
&= \sum\limits_{i=2}^k (1 - H(X_i \mid X_{-i} R_{-i} \Pi_1)) 
			\text{~~~(because $X_i$ is uniform and independent of $X_{-i}$ and of $R_{-i}$)}\\
&\ge \sum\limits_{i=2}^k (1 - H(\XOR_k (X)\mid X_{-i} R_{-i} \Pi_1))
		 \text{~~~(data processing inequality, as $\exists~\Phi \mid X_i=\Phi(\XOR_k(X),X_{-i})$)}\\
&\ge \sum\limits_{i=2}^k (1 - H(\XOR_k(X) \mid X_1 R_1 \Pi_1))\\
&\ge (k-1)(1 - H(\XOR_k(X)\mid X_1 R_1 \Pi_1))\\
&\ge (k-1)(1 - h(\epsilon))\text{~~~~(since player $1$ outputs $\XOR_k(X)$ with error $\epsilon$; see Claim~\ref{cl:entropy_of_calculation})}~.
\end{align*}

\end{proof}

The next theorem follows immediately from Theorem~\ref{thm:micxor} and Theorem~\ref{thm:dsmic}.
\begin{theorem}\label{thm:micxorn}
Let $\mu$ be the uniform distribution on $\{0,1\}^k$. Given any  fixed $0 \leq  \epsilon < \frac{1}{2}$,  for any protocol $\pi$  $\epsilon$-computing $\XOR^n_k$, it holds that
$\MIC_{\mu^n}(\pi) = \Omega(kn)$.
\end{theorem}

We can now prove a lower bound on the communication complexity of $\XOR^n_k$.
Note that the lower bound for $\XOR^n_k$   given in~\cite{KRU} is valid  only for a restricted class of protocols, called
``oblivious'' in~\cite{KRU}.

\begin{theorem}\label{thm:CCXOR2}
Given any fixed $0 \le \epsilon < \frac{1}{2}$, there is a constant $\alpha$ such that for $n\ge\frac{1}{\alpha} k$,
$$\CC^\epsilon(\XOR_k^n) = \Omega(kn)~.$$
\end{theorem}

\begin{proof}
\begin{sloppypar}
Let $\pi$ be a protocol $\epsilon$-computing $\XOR_k^n$.
By Lemma~\ref{lem:CC>=MIC} and Theorem \ref{thm:micxorn}, there exists a constant $\beta$ such that
${\CC(\pi) \ge \beta kn - k^2}$. Let $\alpha < \beta $ be a  constant. For $n\ge\frac{1}{\alpha} k$, we have
$k^2 \le \alpha k n$ and we get $\CC(\pi) \ge (\beta - \alpha) kn = \Omega(kn)$.
\end{sloppypar}

\end{proof}

\section{The function $\AND$}\label{sec:AND}

In this section we consider  an arbitrary $k$-party protocol, $\pi$, where each player has an input bit $x_i$, and
where $\pi$ has to compute the $\AND$ of all the input bits. 
We
prove a lower bound on a certain information theoretic measure  (that we define below) for $\pi$.
\ShortVersion
The proof makes use of the following input distribution.
\ShortVersionEnd
\LongVersion
The proof makes use of a certain input distribution that we will define below.  In the proof we use the
following notations.
Denote by $\overline{1}^t$ the all-$1$ bit-vector of length $t$.
Denote by $\overline{e}^t_{a_1,\dots,a_d}$ the vector obtained from $\overline{1}^t$ by changing the bit $1$ into
the bit $0$ at indexes $a_1,\dots,a_d$. 
To simplify notations, we sometimes omit the  superscript 
$t$ when $t=k$, and write $\overline{e}_{a_1,\dots,a_d}$ or $\overline{1}$.
We further use in the sequel the notation $\delta_{a,b}$ for the Kronecker delta, i.e.,
$\delta_{a,b}=1$ if $a=b$ and $0$ otherwise.
\LongVersionEnd

 \paragraphA{Input distribution.}
Consider the distribution $\mu$
defined as follows. Draw a bit 
$M\sim\Ber(\frac{2}{3},\frac{1}{3})$, and a uniformly random index $Z\in\Ni{1}{k}$. Assign $0$ to $X_Z$. 
If $M=0$, sample $X_{-Z}$ uniformly in $\{0,1\}^{k-1}$; if $M=1$, assign $1^{k-1}$ to $X_{-Z}$. 
We will also work  with the product distribution $\mu^n$. Our distribution is similar to the ones  of 
\cite{BEOPV,CM} in that it  leads to a high information cost (or similar measures) 
for the function $\AND_k$. 
The  distribution that we use has the property that the AND of any input in the support of $\mu$ is $0$.
 This allows us to prove lower bounds for the Disjointness function without the constraint that $k=\Omega(\log n)$ which 
was necessary in~\cite{BEOPV} (but not
in~\cite{CM}).

\LongVersion
Given a protocol $\pi$, let 
$\Pi_i[x_i,m,z]$ denote the distribution of
$\Pi_i$, when the input $X$ is sampled as follows:  $X\sim \mu$, 
conditioned on the fact that $X_i=x_i$, $M=m$ and $Z=z$.
\LongVersionEnd

\LongVersion
\subsection{Basic properties}
We first prove a number of basic properties of $\pi$, under the input distribution $\mu$. The proofs make use of the
general properties of protocols, proved in Section~\ref{sec:multi-party_protocols}.

\paragraphA{Rectangularity.}
We first prove the following lemma, which is an application of Lemma~\ref{lem:rectang} to the specific case of the distribution $\mu$ that we  defined above.
Its proof is given in the appendix.

\begin{lemma}\label{lem:rectangspe}
Let $\pi$ be a  private-coins protocol. 
Let $\calT$ denote the set of possible transcripts of  $\pi$, and for $i\in\Ni{1}{k}$ let $\calT_i$ denote the set of 
possible transcript of by player $i$ so that $\calT\subseteq \calT_1\times\cdots\times\calT_k$.
Then there  exists a function $c:\{0,1\}\times\Ni{1}{k}\times\calT \rightarrow [0,1]$, and for 
every $i\in\Ni{1}{k}$ there is a function $c_i:\{0,1\}\times\Ni{1}{k}\times\calT_i \rightarrow [0,1]$, such that 
$\forall~i\in\Ni{1}{k}$, $\forall~x'\in\{0,1\}$, $\forall~m\in\{0,1\}$, 
$\forall~z\in \Ni{1}{k}\setminus\{i\}$, $\forall~\tau=(\tau_1,\dots,\tau_k)\in\calT$,
$$\Pr[\Pi_i=\tau_i\mid X_i=x',M=m,Z=z]=q_i(x',\tau_i)c_i(m,z,\tau_i)~,$$
 and
$$\Pr[\Pi=\tau\mid X_i=x',M=m,Z=z]=q_i(x',\tau_i)c(m,z,\tau)~.$$
\end{lemma}

\paragraphA{Diagonal lemma.}
The following lemma is a version of Lemma \ref{lem:diag} adapted to our distribution. 
Its proof is given in the appendix.

\begin{lemma}\label{lem:diagspe}
Let $\pi$ be  a private-coins protocol.
For any $i,j\in\Ni{1}{k}$ with $i \neq j$, we have
$h^2(\Pi_i[0,0,j],\Pi_i[1,1,j]) \ge
\frac{1}{2}h^2(\Pi_i(\overline{e}_{i,j}),\Pi_i(\overline{e}_j))$.
\end{lemma}

\paragraphA{Localization.}
The following lemma formalizes the fact that if changing the input of a player changes the transcript of the protocol, then 
this change necessarily appears in the partial transcript of that player. For randomized protocols this change is observed and 
quantified by the Hellinger distance between the distributions of the transcripts.
The proof is given in the appendix.

\begin{lemma}\label{lem:localspe}
Let $\pi$ be a  private-coins protocol.
$\forall i\in\Ni{1}{k},~\forall j\in\Ni{1}{k}\setminus\{i\}$,
$$h(\Pi_{i}(\overline{e}_{i,j}),\Pi_{i}(\overline{e}_{j})) = h(\Pi(\overline{e}_{i,j}),\Pi(\overline{e}_{j}))~.$$
\end{lemma}

\LongVersionEnd

\subsection{Switched multi-party information cost of $\AND_k$}\label{subsec:SICAND}

We propose the following definition, which is an adaptation of the switched information cost of \cite{BEOPV}. 
We call it Switched Multi-party Information Cost ($\SIC$).
\begin{definition}
For a $k$-player protocol $\pi$ with inputs drawn from $\mu^n$ let 
$$\SIC_{\mu^n}(\pi)=\sum\limits_{i=1}^k \left(I(X_i ; \Pi_i \mid M Z) + I(M;\Pi_i \mid X_i Z)\right)~.$$
\end{definition}

Note that the notion of $\SIC$ is only defined with respect to the distribution $\mu^n$ that we  defined,
 and we may thus omit the distribution from the notation. 
 We note that in order to simplify the expressions we often consider the  public randomness as implicit
in the information theoretic expressions we use below. It can be materialized  either as part of the transcript or
 in the conditioning of the information theoretic expressions.

We can now prove the main result of this section.

\begin{theorem}\label{prop:SICAND}
For any fixed $0 \leq  \epsilon < \frac{1}{2}$,  for any protocol $\pi$ externally $\epsilon$-computing $\AND_k$, 
$${\SIC_\mu(\pi) = \Omega(k)}~.$$
\end{theorem}

\begin{proof}
We prove below the claim for an arbitrary  private-coins protocol  $\pi$. The claim for general protocols
(i.e., with public randomness)
 then follows from averaging over all possible assignments to the public randomness.  
  
Observe that by the definition of $\mu$, for any $i\in\Ni{1}{k}$, if  $M=0$ and $Z=z\neq i $, then
$ X_i \sim \Ber(\frac{1}{2},\frac{1}{2})$. We therefore get by Lemma \ref{lem:ICandHel} that
\begin{equation}
\label{ineq:helinger_1}
\forall i\in\Ni{1}{k}, ~\forall z\in\Ni{1}{k}\setminus\{i\}, ~~
I(X_i;\Pi_i \mid M=0,Z=z) \ge h^2(\Pi_i[0,0,z],\Pi_i[1,0,z])~.
\end{equation}

Similarly, by the definition of $\mu$  
we have that  for any $i\in\Ni{1}{k}$,
 if
$X_i=1$ and $Z=z\neq i $, then $M\sim\Ber(\frac{1}{2},\frac{1}{2})$,  
and we get by Lemma \ref{lem:ICandHel} that 
\begin{equation}
\label{ineq:helinger_2}
\forall i\in\Ni{1}{k},~  \forall z\in\Ni{1}{k}\setminus\{i\}, ~~
I(M;\Pi_i \mid X_i=1,Z=z) \ge h^2(\Pi_i[1,0,z],\Pi_i[1,1,z])~.
\end{equation}

Let us now define $\SIC_i(\pi) = I(X_i ; \Pi_i \mid M Z) + I(M;\Pi_i \mid X_i Z)$, so that 
$\SIC(\pi) = \sum\limits_{i=1}^k \SIC_i(\pi)$.
We get
\begin{align*}
\SIC_i(\pi) &= I(X_i ; \Pi_i \mid M Z) + I(M;\Pi_i \mid X_i Z)\\
&= \Exp\limits_z \left[I(X_i ; \Pi_i \mid M, Z=z) + I(M;\Pi_i \mid X_i, Z=z)\right]\\
&\ge \frac{1}{k}\sum\limits_{z\neq i} \left[I(X_i ; \Pi_i \mid M, Z=z) + I(M;\Pi_i \mid X_i, Z=z)
\right]\\
&\ge \frac{1}{k}\sum\limits_{z\neq i}\left[\Pr[M=0 \mid Z=z]I(X_i ; \Pi_i \mid M=0, Z=z)\right.+\\  
&\left.~~~~~~~~~~~~\Pr[X_i=1 \mid Z=z] I(M;\Pi_i \mid X_i=1, Z=z) \right]~.
\end{align*}

By the definition of $\mu$, $\Pr[M=0 \mid Z=z] = \frac{2}{3}$ for any $z$. Also, for any $i\neq z$,
\begin{align*}
\Pr[X_i=1 \mid Z=z] &= \Pr[M=0 \mid Z=z]\Pr[X_i=1 \mid M=0, Z=z]\\
&\,+\Pr[M=1 \mid Z=z]\Pr[X_i=1 \mid M=1, Z=z]\\
&=\frac{2}{3}\cdot\frac{1}{2} + \frac{1}{3}\cdot 1 = \frac{2}{3}~.
\end{align*}

Thus, using Inequalities~(\ref{ineq:helinger_1}) and~(\ref{ineq:helinger_2}), we have
\begin{align*}
\SIC_i(\pi) &\ge \frac{1}{k}\sum\limits_{z\neq i}\left[\frac{2}{3} h^2(\Pi_i[0,0,z],\Pi_i[1,0,z]) + 
 \frac{2}{3} h^2(\Pi_i[1,0,z],\Pi_i[1,1,z]) \right]\\
&\ge\frac{1}{3k}\sum\limits_{z\neq i}\left[ h(\Pi_i[0,0,z],\Pi_i[1,0,z]) + h(\Pi_i[1,0,z],\Pi_i[1,1,z])\right]^2\\
&\ge\frac{1}{3k}\sum\limits_{z\neq i}h^2(\Pi_i[0,0,z],\Pi_i[1,1,z]) \text{~~~(by the triangular inequality).}
\end{align*}

We have
\begin{align*}
\SIC(\pi)& = \sum\limits_{i=1}^k \SIC_i(\pi)\\
&\ge\frac{1}{3k}\sum\limits_{i,z \mid i\neq z}h^2(\Pi_i[0,0,z],\Pi_i[1,1,z])\\
&\ge\frac{1}{3k}\sum\limits_{\{i,z\}}
[h^2(\Pi_i[0,0,z],\Pi_i[1,1,z]) + h^2(\Pi_z[0,0,i],\Pi_z[1,1,i])]\\
&\ge \frac{1}{6k}\sum\limits_{\{i,z\}} 
[h^2(\Pi_{i}(\overline{e}_{i,z}),\Pi_{i}(\overline{e}_{z})) +
h^2(\Pi_{z}(\overline{e}_{i,z}),\Pi_{z}(\overline{e}_{i}))] \text{~~~(by Lemma \ref{lem:diagspe})}\\
&\ge \frac{1}{6k}\sum\limits_{\{i,z\}} 
[h^2(\Pi(\overline{e}_{i,z}),\Pi(\overline{e}_{z})) +
h^2(\Pi(\overline{e}_{i,z}),\Pi(\overline{e}_{i}))]\text{~~~(by Lemma \ref{lem:localspe})}\\
&\ge \frac{1}{12k}\sum\limits_{\{i,z\}} 
[h(\Pi(\overline{e}_{i,z}),\Pi(\overline{e}_{z})) +
h(\Pi(\overline{e}_{i,z}),\Pi(\overline{e}_{i}))]^2\\
&\ge \frac{1}{12k}\sum\limits_{\{i,z\}} 
h^2(\Pi(\overline{e}_{i}),\Pi(\overline{e}_{z}))\text{~~~(by the triangular inequality)}\\
&\ge \frac{1}{24k}\sum\limits_{\{i,z\}} 
h^2(\Pi(\overline{e}_{i}),\Pi(\overline{1}))
\text{~~(by Lemma \ref{lem:diag}, omitting part of the right hand side term)}\\
&\ge \frac{1}{24k}\sum\limits_{\{i,z\}} \frac{(1-2\epsilon)^2}{2}\text{~~~(by Lemma \ref{lem:HelErr})}\\
&\ge \frac{(k-1)(1-2\epsilon)^2}{96} = \Omega(k)~.
\end{align*}
\end{proof}

\section{The function Disjointness}\label{sec:Disj}

In the $k$ players $n$-bit disjointness function $\Disj_k^n$, every player $i\in\Ni{1}{k}$ has an $n$-bit string 
$(x_i^{\ell})_{\ell\in\Ni{1}{n}}$, and the players have to output $1$ if and only if there exists a coordinate $\ell$ where 
all players have the bit  $1$. Formally, 
$\Disj_k^n(x) = \bigvee_{\ell=1}^n \bigwedge_{i=1}^k x_i^{\ell}$.

\subsection{Switched multi-party information cost of $\Disj_k^n$}\label{subsec:SICDISJ}

We first prove a direct-sum-type property which  allows us to make the link between the 
functions $\AND_k$ and $\Disj_k^n$. A similar property was proved
in~\cite{BEOPV} in the coordinator model; our peer-to-peer model requires a different,  more 
 involved, construction, since we do not have the coordinator, and moreover no player can act as the coordinator since
  it would get too  much information.  
  Since  $\Disj_k^n$ is  the disjunction of $n$ $\AND_k$ functions, we 
  analyze the switched multi-party information cost of $\Disj_k^n$ using the  distribution $\mu^n$.

\begin{lemma}\label{lem:DSSIC}
Let $k>3$. For any protocol $\pi$ externally $\epsilon$-computing $\Disj_k^n$, 
there exists a protocol $\pi'$ externally $\epsilon$-computing $\AND_k$ such that
$$\SIC_{\mu^n}(\pi) \ge n \cdot \SIC_{\mu}(\pi')~.$$
\end{lemma}

\begin{proof}
Based on an arbitrary protocol $\pi$ for $\Disj_k^n$, we define a protocol $\pi'$ for $\AND_k$, and then analyze 
$\SIC_{\mu^n}(\pi)$  and  $\SIC_{\mu}(\pi')$.
Let $u \in \{0,1\}^k$ be the input to $\pi'$ such that $u_i$ is given to player $i$.
We note that we cannot use a protocol similar to the one used in~\cite{BEOPV} since in the peer-to-peer setting one
does not have a coordinator that can sample the inputs for the player. We thus need to sample the inputs
in a distributed way, while keeping the information complexity under control using  classic secret sharing techniques.
The protocol $\pi'$ is defined as follows. 

\begin{enumerate}
\item
The players first sample publicly an index $L$ uniformly in $\Ni{1}{n}$, and then sample 
publicly $Z^t$, for ${t\in\Ni{1}{n}\setminus\{L\}}$, independently and uniformly in $\Ni{1}{k}$.
\item
They then proceed to sample $M^t$, for ${t\in\Ni{1}{n}\setminus\{L\}}$, as follows.
The set of players is partitioned  into two subsets, $\{1,2\}$ and $\{3,\dots,k\}$.
Player $1$ samples $M^1\dots M^{L-1}$ and sends the sampled values to player $2$
(player $3$ samples $M^{L+1}\dots M^n$, see below).
\item
Then Player $1$ samples $X_1^1\dots X_1^{L-1}$ according to the distribution $\mu$, 
and player $2$ samples $X_2^1\dots X_2^{L-1}$, according to the distribution $\mu$. 
Observe that they can do this as they know $M^1,\dots,M^{L-1},Z^1,\dots,Z^{L-1}$.
\item
Players $1$ and $2$ then apply the following procedure to communicate $X_j^t$ to player $j$, for $j>2$ and $t<L$:
Player $1$ sends a bit $p_j^t$ to player $j$,  and sends a bit  $v_j^t$ to player $2$. Player $2$  then sends a bit $q_j^t$ to player $j$.
 Player $j$  then defines
$X_j^t=p_j^t\oplus q_j^t$. The bits $p_j^t$, $q_j^t$ and  $v _j^t$ are generated in the following way.
\begin{itemize}
\item If $Z^t=j$ player $1$ privately samples a random bit $v_j^t$. It then sets $p_j^t=v_j^t$. Player $2$
sets $q_j^t=v_j^t$. Player $j$ thus defines $X_j^t=0$.
\item If $Z^t\neq j$ and $M^t=0$, player $1$ privately samples two independent random  bits $p_j^t$ and $v_j^t$.
Player $2$ privately samples a random bit $q_j^t$. The bit $X_j^t$ defined by player $j$ is in this case a uniform random bit.
Note that it is not necessary for the correctness of the protocol that bit $v_j^t$ is sent to Player 2 in this case; it is sent here
only to make our notations simpler.
\item If $Z^t\neq j$ and $M^t=1$, player $1$ privately samples a random bit $v_j^t$. It then sets $p_j^t=v_j^t$. Player $2$  defines $q_j^t=v_j^t\oplus 1$. Player $j$ thus defines $X_j^t=1$.
\end{itemize}
\item
Player $3$ samples $M^{L+1}\dots M^n$ and sends the sampled values to players $4$ to $k$. Every player 
$i\ge 3$ privately samples $X_i^{L+1}\dots X_i^n$.
\item
Players $3$ and $4$ then apply the same procedure as players $1$ and $2$, in order 
to communicate $X_j^t$ to player $j$, for $j\le 2$ and $t>L$. We denote by $p_1^t$ and by $p_2^t$ the bits sent
by player $3$ to player $1$ and to player $2$, respectively; by $q_1^t$ and by $q_2^t$ the bits sent by player $4$ to 
player $1$ and to player $2$, respectively; and by $v_1^t$ and by $v_2^t$ the bits sent by player $3$ to player $4$.
\item 
\label{step:prot_pi}
Now all the players run protocol $\pi$, on the input composed of (1) the values defined above for 
$x_i^t$, $i \in \Ni{1}{k}$, $t \in \Ni{1}{n} \setminus \{ L \}$ , and (2) $x_i^{L}=u_i$, for $i \in \Ni{1}{k}$.
\item The output of the protocol $\pi'$ is the output of the protocol $\pi$.
Observe that one can define the function needed as in the definition of an externally $\epsilon$-computing protocol, 
since the transcript of $\pi'$ is composed of a fixed-length initial portion and then the transcript of $\pi$.
\end{enumerate}

First observe that if $\pi$ computes $\Disj_k^n$ with error $\epsilon$, then $\pi'$ computes $\AND_k$ with error $\epsilon$,
and this is regardless of the values of the random bits used in the construction of  the input to $\pi$
(this property of the distribution of the input to $\pi$ is called \textit{collapsing on  coordinate $L$} in, e.g.,  ~\cite{BEOPV}).

Now observe that if the input to protocol $\pi'$, denote it $U$, is distributed according to $\mu$ (as defined above)
 then the definition of $\pi'$ guarantees that the input to protocol $\pi$,  $X$, is distributed 
 according to $\mu^n$.
Using the notation we use for $\mu$ we can write that  if  $(U,N,S) \sim\mu$ then $(X,M,Z)\sim\mu^n$.

We now give an upper bound on $\SIC_{\mu}(\pi')$ in terms of $\SIC_{\mu^n}(\pi)$.
To this end we first express the transcripts of protocol $\pi'$, $\Pi'_i$, $1 \leq i \leq k$, in terms of the transcripts 
$(\Pi_i)$ of the protocol $\pi$, run in Step~\ref{step:prot_pi}.

Let us take player $2$ and express $\Pi'_2$ as a 
 function of $\Pi_2$. 
 Taking into account the preliminary sampling procedure of protocol $\pi'$, we can write $\Pi'_2$ in four parts.
\begin{enumerate}
\item The values which are a function of the public randomness  used by 
$\pi'$: $L,Z^{-L}$ (for simplicity we include 
 the sampled values and not the random bits).
\item
\begin{itemize}
\item Read by player $2$ (and sent by player $1$), $M^{<L}$.
\item  Read by player $2$ (and sent by player $1$), all the $v_j^t$ for $j>2, t<L$ (denoted below as $v_{>2}^{<L}$).
\item  Sent by player $2$, all the $q_j^{t}$, for $j>2, t<L$ (denoted below as $q_{>2}^{<L}$).
\end{itemize}
\item Player $2$ also receives $p_2^{L+1}\dots p_2^n,q_2^{L+1}\dots q_2^n$ from players $3$ and $4$ (denoted below 
as $p_2^{>L}$ and $q_2^{>L}$).
\item The last part is the transcript  of player $2$ when running $\pi$.
\end{enumerate}

Thus, the transcript $\Pi'_2$ can be written as $L Z^{-L}M^{<L} v_{>2}^{<L}  q_{>2}^{<L}p_2^{>L}q_2^{>L}\Pi_2$. However, 
in the manipulations of $\SIC$ we can write $\Pi'_2$ also as $Z^{-L}M^{<L}X_2^{>L}\Pi_2$. This is because
\begin{align*}
  I(U_2;\Pi'_2 \mid NS) + I(N;\Pi'_2 \mid U_2 S) 
    & =      I(U_2; L Z^{-L}M^{<L} v_{>2}^{<L}  q_{>2}^{<L}p_2^{>L}q_2^{>L}\Pi_2\mid NS)\,+\\
    &~~~~I(N; L Z^{-L}M^{<L} v_{>2}^{<L}  q_{>2}^{<L}p_2^{>L}q_2^{>L}\Pi_2\mid U_2 S) \\ 
    &=   I(U_2; L Z^{-L}M^{<L} v_{>2}^{<L}  q_{>2}^{<L}X_2^{>L}\Pi_2\mid NS)\,+\\
    &~~~~I(N; L Z^{-L}M^{<L} v_{>2}^{<L}  q_{>2}^{<L}X_2^{>L}\Pi_2\mid U_2 S) \\ 
    &  =   I(U_2;Z^{-L}M^{<L}X_2^{>L}\Pi_2 \mid NS) 
           +   I(U_2;L  v_{>2}^{<L}  q_{>2}^{<L} \mid NSZ^{-L}M^{<L}X_2^{>L}\Pi_2)\,+\\ 
    &~~~~I(N; Z^{-L}M^{<L}X_2^{>L}\Pi_2 \mid U_2 S) 
             +    I(N;  L  v_{>2}^{<L}  q_{>2}^{<L} \mid U_2 SZ^{-L}M^{<L}X_2^{>L}\Pi_2 )  \\  
    & =  I(U_2; Z^{-L}M^{<L}X_2^{>L}\Pi_2 \mid NS) + I(N;Z^{-L}M^{<L}X_2^{>L}\Pi_2 \mid U_2 S) ~,
\end{align*}
where the second equality follows from the fact that the distribution of $p_2^{t}q_2^{t}$, for all
$t > L$, is uniform for $p_2^{t}\oplus q_2^{t}=x_2^{t}$ and independent 
of $U_2$ (resp., of N), conditioned on  $X_2^{>L}$, the rest of the 
transcript $\Pi'_2$,  and $S$ (resp., $U_2$); the third equality follows from the chain 
rule; and the last equality follows from the fact that $I(U_2;L  v_{>2}^{<L}  q_{>2}^{<L} \mid NSZ^{-L}M^{<L}X_2^{>L}\Pi_2)=0$
 and ${I(N;  L  v_{>2}^{<L}  q_{>2}^{<L} \mid U_2 SZ^{-L}M^{<L}X_2^{>L}\Pi_2 ) =0}$.
These last two equations follow from
the fact that $Lv_{>2}^{<L}  q_{>2}^{<L}$ is independent of $U_2$ (resp.,  of $N$), even conditioned on $SZ^{-L}M^{<L}X_2^{>L}\Pi_2$ and on $N$ (resp., on $U_2$). 
\adifuture{the last argument should be improved. Could we write for ourselves a full proof ?}

By similar argument we can write, in the manipulations of $\SIC$,  $\Pi'_1$ as $Z^{-L}M^{<L}X_1^{>L}\Pi_1$, 
and for $i\ge 3$, $\Pi'_i$ as $Z^{-L}M^{>L}X_i^{<L}\Pi_i$.

\vspace{0.25cm}
We have
\begin{align*}
\SIC_\mu(\pi') &= \sum\limits_{i=1}^k \left(I(U_i;\Pi'_i \mid NS) + I(N;\Pi'_i \mid U_i S) \right)\\
&=\Exp_{\ell} \left[
\sum\limits_{i=1}^2\left(I(X_i^\ell ; Z^{-\ell}M^{<\ell}X_i^{>\ell}\Pi_i \mid M^\ell Z^\ell) + 
I(M^\ell;Z^{-\ell}M^{<\ell}X_i^{>\ell}\Pi_i \mid X_i^\ell Z^\ell)\right) \right.\\
&~~~~+\left.\sum\limits_{i=3}^k\left(I(X_i^\ell ; Z^{-\ell}M^{>\ell}X_i^{<\ell}\Pi_i \mid M^\ell Z^\ell )
+I(M^\ell;Z^{-\ell}M^{>\ell}X_i^{<\ell}\Pi_i \mid X_i^\ell Z^\ell)\right)
\vphantom{\sum\limits_{i=1}^2}\right]\\
&=\Exp_\ell \left[
\sum\limits_{i=1}^2\left(I(X_i^\ell ; \Pi_i \mid X_i^{>\ell} M^{\le \ell} Z) + 
I(M^\ell;\Pi_i \mid X_i^{\ge \ell} M^{<\ell} Z)\right) \right.\\
&~~~~+\left.\sum\limits_{i=3}^k\left(I(X_i^\ell ; \Pi_i \mid X_i^{<\ell} M^{\ge \ell} Z)
+I(M^\ell;\Pi_i \mid X_i^{\le \ell} M^{>\ell} Z)\right)
\vphantom{\sum\limits_{i=1}^2}\right]~.
\end{align*}

Now, applying Lemma \ref{lem:Bra2.10}, we have that for any $\ell$
\begin{align*}
I(X_i^\ell ; \Pi_i \mid X_i^{>\ell} M^{\le \ell} Z) &\le I(X_i^\ell ; \Pi_i \mid X_i^{>\ell} M Z) ~~
\text{ (since }I(X_i^\ell ; M^{>\ell} \mid X_i^{>\ell} M^{\le \ell} Z)=0\text{),}\\
I(M^\ell;\Pi_i \mid X_i^{\ge \ell} M^{<\ell} Z) &\le I(M^\ell;\Pi_i \mid X_i M^{<\ell} Z) ~~
\text{ (since }I(M^\ell;X_i^{<\ell} \mid X_i^{\ge \ell} M^{<\ell} Z)=0\text{),}\\
I(X_i^\ell ; \Pi_i \mid X_i^{<\ell} M^{\ge \ell} Z) &\le I(X_i^\ell ; \Pi_i \mid X_i^{<\ell} M Z) ~~
\text{ (since }I(X_i^\ell ; M^{<\ell} \mid X_i^{<\ell} M^{\ge \ell} Z)=0\text{),}\\
I(M^\ell;\Pi_i \mid X_i^{\le \ell} M^{>\ell} Z) &\le I(M^\ell;\Pi_i \mid X_i M^{>\ell} Z) ~~
\text{ (since }I(M^\ell;X_i^{>\ell} \mid X_i^{\le \ell} M^{>\ell} Z)=0\text{).}
\end{align*}

Thus
\begin{align*}
\SIC_\mu(\pi') &\le \Exp_\ell \left[
\sum\limits_{i=1}^2\left(I(X_i^\ell ; \Pi_i \mid X_i^{>\ell} M Z) +
I(M^\ell;\Pi_i \mid X_i M^{<\ell} Z)\right) \right.\\
&~~~+\left.\sum\limits_{i=3}^k\left(I(X_i^\ell ; \Pi_i \mid X_i^{<\ell} M Z)
+I(M^\ell;\Pi_i \mid X_i M^{>\ell} Z)\right)
\vphantom{\sum\limits_{i=1}^2}\right]\\
&\le \frac{1}{n}\sum\limits_{\ell=1}^n \left[
\sum\limits_{i=1}^2\left(I(X_i^\ell ; \Pi_i \mid X_i^{>\ell} M Z) + 
I(M^\ell;\Pi_i \mid X_i M^{<\ell} Z)\right) \right.\\
&~~~+\left.\sum\limits_{i=3}^k\left(I(X_i^\ell ; \Pi_i \mid X_i^{<\ell} M Z)
+I(M^\ell;\Pi_i \mid X_i M^{>\ell} Z)\right)
\vphantom{\frac{1}{n}}\right]\\
&\le \frac{1}{n}\left[
\sum\limits_{i=1}^2\left(\sum\limits_{\ell=n}^1 I(X_i^\ell ; \Pi_i \mid X_i^{>\ell} M Z) + 
\sum\limits_{\ell=1}^n I(M^\ell;\Pi_i \mid X_i M^{<\ell} Z)\right) \right.\\
&~~~~~+\left.\sum\limits_{i=3}^k\left(\sum\limits_{\ell=1}^n I(X_i^\ell ; \Pi_i \mid X_i^{<\ell} M Z)
+\sum\limits_{\ell=n}^1 I(M^\ell;\Pi_i \mid X_i M^{>\ell} Z)\right)
\vphantom{\frac{1}{n}}\right]\\
&\le \frac{1}{n}\left[
\sum\limits_{i=1}^2\left( I(X_i ; \Pi_i \mid M Z) + 
I(M;\Pi_i \mid X_i Z)\right) \right.\\
&~~~~~+\left.\sum\limits_{i=3}^k\left(I(X_i ; \Pi_i \mid M Z)
+ I(M;\Pi_i \mid X_i Z)\right)
\vphantom{\frac{1}{n}}\right]\\
&\le \frac{1}{n}\sum\limits_{i=1}^k\left( I(X_i ; \Pi_i \mid M Z) + 
I(M;\Pi_i \mid X_i Z)\right)\\
&\le \frac{1}{n} \SIC_{\mu^n}(\pi)~.
\end{align*}
\end{proof}

Coupled with the lower bound on $\SIC(\pi')$ for any protocol $\pi'$  that computes $\AND_k$ 
(Section \ref{sec:AND}), the above lemma gives us a lower bound on $\SIC(\pi)$ for any protocol that computes  the function $\Disj_k^n\,$:

\begin{theorem}\label{thm:SIC>=kn}
Let $k>3$.
Given any fixed  $0 \leq \epsilon < \frac{1}{2}$,  for any protocol $\pi$ externally $\epsilon$-computing $\Disj_k^n$ it holds that 
$$\SIC_{\mu^n}(\pi) = \Omega(kn)~.$$
\end{theorem}

\subsection{Multi-party information complexity and communication complexity of $\Disj_k^n$}

\ShortVersion
The next  lemma is key to our argument. The theorem that follows is a consequence of it and of Theorem~\ref{thm:SIC>=kn}.
\ShortVersionEnd
\LongVersion
We now prove a lemma that will allow us to obtain a lower bound on the multi-party peer-to-peer 
communication complexity of the disjointness function.
\LongVersionEnd

\begin{lemma}\label{lem:SMIC<MIC}
For any $k$-player protocol $\pi$, 
$\SIC_{\mu^n}(\pi) \le \MIC_{\mu^n}(\pi)$.
\end{lemma}

\begin{proof}
We first prove that $$\forall~i\in\Ni{1}{k},~~ I(M;\Pi_i \mid X_i R_i Z) \le I(X_{-i} ; \Pi_i \mid X_i R_i)~.$$

\begin{align*}
I(M;\Pi_i \mid X_i R_i Z) &\le I(M X_{-i};\Pi_i  \mid X_i R_i Z) \\
&= I(X_{-i};\Pi_i \mid X_i R_i Z) + I(M ;\Pi_i \mid X R_i Z) \text{~~~(chain rule)}\\
&\le I(X_{-i};\Pi_i \mid X_i R_i Z) + I(M ;\Pi_i R_{-i} \mid X R_i Z)\\
&=  I(X_{-i};\Pi_i \mid X_i R_i Z) + I(M ;R_{-i} \mid X R_i Z) + I(M ; \Pi_i \mid X R Z) \text{~~~(chain rule)}\\
&= I(X_{-i};\Pi_i \mid X_i R_i Z) + I(M ; \Pi_i \mid X R Z)\\ 
&\le I(X_{-i};\Pi_i \mid X_i R_i Z) + H(\Pi_i \mid X R Z)\\
&= I(X_{-i};\Pi_i \mid X_i R_i Z) \text{~~~(because $XR$ determines $\Pi_i$)}\\
&= H(\Pi_i \mid X_i R_i Z) - H(\Pi_i \mid X R_i Z)\\
&=H(\Pi_i \mid X_i R_i Z) - H(\Pi_i \mid X R_i)\\
&\le H(\Pi_i \mid X_i R_i ) - H(\Pi_i \mid X R_i)  \\
&\le I(X_{-i};\Pi_i \mid X_i R_i)~.
\end{align*}

We now prove that
$$\forall~i\in\Ni{1}{k}, ~~I(X_i ; \Pi_i \mid M Z) \le I(X_i ; \Pi_i \mid X_{-i} R_{-i}).$$

Since  by the definition of $\mu$ $I(X_i ; X_{-i} R_{-i} \mid MZ) = 0$, we get by Lemma \ref{lem:Bra2.10} that
$$I(X_i ; \Pi_i \mid M Z) \le I(X_i ; \Pi_i \mid X_{-i} R_{-i} M Z)~,$$
and
\begin{align*}
I(X_i ; \Pi_i \mid X_{-i} R_{-i} M Z) &= H(\Pi_i \mid X_{-i} R_{-i} M Z) - H(\Pi_i \mid X R_{-i} M Z)\\
&= H(\Pi_i \mid X_{-i} R_{-i} M Z) - H(\Pi_i \mid X R_{-i})\\
&\le H(\Pi_i \mid X_{-i} R_{-i}) - H(\Pi_i \mid X R_{-i})\\
&= I(X_i ; \Pi_i \mid X_{-i} R_{-i})~.
\end{align*}
Thus we have
$$I(X_i ; \Pi_i \mid M Z) \le I(X_i ; \Pi_i \mid X_{-i} R_{-i})~.$$
Summing over $i\in\Ni{1}{k}$ concludes the proof.
\end{proof}

The next theorem follows immediately from Theorem~\ref{thm:SIC>=kn} and Lemma~\ref{lem:SMIC<MIC}.

\begin{theorem}\label{thm:MICDisj}
Let  $k>3$.  Given any fixed  $0 \leq \epsilon <\frac{1}{2}$, for any protocol $\pi$ externally 
$\epsilon$-computing $\Disj_k^n$, it holds that
 $$\MIC_{\mu^n}(\pi) = \Omega(kn)~.$$
\end{theorem}

We   now conclude with a lower bound on  the randomized communication complexity of the disjointness function.
\begin{theorem}
Given any  fixed $0 \leq \epsilon <\frac{1}{2}$, there is a constant $\alpha$ such that for $n\ge\frac{1}{\alpha} k$,
$$\CC^\epsilon(\Disj_k^n) = \Omega(kn)~.$$
\end{theorem}

\begin{proof}
For $k=3$ the theorem follows from the fact that $\CC^\epsilon(\Disj_3^n) \geq \CC^\epsilon(\Disj_2^n)$
(simply by letting Alice simulate internally a third player with an all-$1$ input), and from  
$\CC^\epsilon(\Disj_2^n) =\Omega(n)$ (cf.~\cite{CP}). 

Assume now that $k>3$.
Let $\pi$ be a protocol $\epsilon$-computing $\Disj_k^n$. We first convert $\pi$ into a protocol $\pi'$ which
{\em externally} $\epsilon$-computes $\Disj_k^n$.
The protocol $\pi'$ is defined as follows. 
For every bit $b$ sent by a player in $\pi$, the same player sends in $\pi'$ two bits  $b.b$. 
In addition, in $\pi'$, when player $1$ stops and returns its output, it sends
to player $2$ the message $b.(1-b)$, where $b$ is the output it computed.

Since in $\pi$ player $1$ $\epsilon$-computes the function $\Disj_k^n$, $\pi'$ externally $\epsilon$-computes $\Disj_k^n$. 
Observe that $\CC(\pi')= 2\cdot \CC(\pi) + 2$.

By Lemma~\ref{lem:CC>=MIC} and Theorem \ref{thm:MICDisj}, there exists a constant $\beta$ such that
$\CC(\pi') \ge \beta kn - k^2 $. Let $\alpha < \beta $ be a constant. For $n\ge\frac{1}{\alpha} k$, we have
$k^2 \le \alpha k n$ and we get $\CC(\pi') \ge (\beta - \alpha) kn = \Omega(kn)$, and  $\CC(\pi) = \Omega(kn)$.

\end{proof}

We  note that our tight lower bound holds also for protocols where only one player is required to output the 
value of the function. 

\section{Randomness complexity of private protocols}\label{sec:Rand}

\ShortVersion
A  protocol $\pi$ is said to {\em privately} compute a  function  $f$ if, at the end of the execution of the protocol, 
the players have learned nothing but the value of that function. 
We now prove that  the (information theoretic) private computation of $\Disj_k^n$ requires $\Omega(n)$ random bits.
We prove this result using  the 
information theoretic results for $\Disj_k^n$ of  the previous sections.
The definitions and the details of the proof are deferred to the  full version of the paper.
\begin{theorem}
Let $k>3$. Then $\calR(\Disj_k^n) = \Omega(n)$,  where $\calR(f)$  is the minimum  number of random bits necessary for
 a protocol to privately compute $f$. 
\end{theorem}
\ShortVersionEnd

\LongVersion
In this section we give a lower bound of $\Omega(n)$ on the (information theoretic private computation) randomness complexity 
of the function 
$\Disj_k^n$, i.e., we prove that in order to privately compute $\Disj_k^n$ one needs $\Omega(n)$ random bits.
The significance  of this result lies in that it is the first such lower bound that grows with the size of the
input, which is $kn$, while the  output  remains a single bit.

\subsection{Private protocols and randomness}

A  protocol $\pi$ is said to {\em privately} compute a given function if, at the end of the execution of the protocol, the 
players have learned nothing but the value of that function.  We note that the literature 
devoted to private computation 
usually focuses on $0$-error protocols, and therefore, in the rest of this section, we will restrict ourselves to the 
case of  $0$-error protocols. 

Furthermore, the literature on private computation is  focused  on  \textit{synchronous} protocols. In what follows we 
therefore only consider protocols in that setting. 
In the synchronous setting, protocols advance according to a global round structure. At every round, each player sends a message to every other player. In addition, each player has an output tape.
In order to ensure that no player is ever engaged in an infinite computation process, it is required that on any input and randomness assignment, every player eventually stops sending messages. 
That is, for a synchronous  protocol $\pi$ let $t_i(x,r)$ be the smallest integer such that if $\pi$ is run on 
$(x,r)$ then player $i$ does not send any message and does not write on its output tape after round $t_i(x,r)$. If no such integer exists then 
$t_i(x,r)=\infty$. The requirement is that  for every player $i$, input $x$, and randomness assignment $r$  $t_i(x,r) < \infty$.

\medskip
The following lemma, which is a consequence of K\"{o}nig's lemma (cf.~\cite{Kleene}), applies to any synchronous protocol.
\begin{lemma}\label{lemma:alaKonig}
Let $\pi$ be  a synchronous  protocol. If for any $i$, $x$, and $r$, $t_i(x,r) < \infty$, then 
there exists an integer $t_f$ such that for any $i$, $x$, and $r$, $t_i(x,r) < t_f $.
\end{lemma}

Based on the above lemma one can transform any synchronous protocol into a protocol that always runs in a fixed number
of rounds, and where all players output at the protocol's end. This is done by simply delaying the output until 
round  $t_f$. Observe that such transformation does not change the 
transcript of the protocol or any other measure such as the number of random bits used.

We can now formally define privacy:
\begin{definition} \label{def:privacy}
A $k$-player  protocol $\pi$ computing a function $f$ is \textit{private}\footnote{In this paper we consider 
only the setting of $1$-privacy, which we call here for simplicity, privacy.} if 
for every player $i \in \Ni{1}{k}$,
for all pairs of inputs $x=(x_1,\dots,x_k)$ and $x'=(x'_1,\dots,x'_k)$, such that ${f(x) = f(x')}$ and $x_i = x'_i$,
for all possible private random assignments $r_i$ of player $i$, and all possible public random assignments $r^p$, 
it holds that for any transcript $T$ 
$$\Pr[\Pi_i = T \mid R_i=r_i~;~X=x~;~R^p=r^p] = \Pr[\Pi_i = T \mid R_i=r_i~;~X=x'~;~R^p=r^p]$$ where 
the probability is over the randomness $R_{-i}$, and where $\Pi_i$ is the sequence of all messages sent to 
player $i$.
\end{definition}

It is well known that in the multi-party case, i.e., when we have $k \geq 3$ players, any function 
can be computed privately in the peer-to-peer model \cite{BGW,CCD}.
Private protocols require the players to make use of their private randomness. The minimal amount of 
private randomness needed to design a private protocol for a given function is referred to as the 
\textit{randomness complexity} of that function. 
While in the present paper we make use of the notion of entropy, many papers on  randomness  in private protocols 
make use of the notion of the number of random bits in order to measure ``the amount of randomness used''.
 We repeat here the definitions used in those papers.

\begin{definition}
A communication protocol is said to be $d$-random if, on any run, the total number of private random bits
 used by all the players is at most $d$.
\end{definition}

\begin{definition}
The randomness complexity $\calR(f)$ of a function $f$ is the minimal integer $d$ such that there exists a $d$-random 
private protocol computing $f$.
\end{definition}

We will also use the following two (finer) notions which in fact make use of the notion of entropy.

\begin{definition}
The randomness complexity of a protocol $\pi$ on input distribution $\eta$ is defined as
$$\calR_\eta(\pi) = H(\Pi \mid X R^p)~.$$
\end{definition}

\begin{definition}
The randomness complexity of a function $f$ on input distribution $\eta$ is defined as
$$\calR_\eta(f) = \inf\limits_{\pi\text{ private  protocol computing }f}\calR_\eta(\pi)~.$$
\end{definition}

Once the input and the public coins are fixed, the entropy of the transcript of a protocol comes solely from the private randomness. 
Thus, for any input distribution $\eta$, $\calR_\eta(\pi)$ provides a lower bound on the entropy of the private randomness used 
by all the players in the protocol $\pi$.
In order to relate our results (which are stated in terms of entropy) to the notions previously used in the literature on the 
analysis of  randomness in private protocols, we use the fact that, up to constant factors, the number of  (uniform)
random bits necessary for the generation of a random variable with a given entropy is equal to that entropy
 (cf.~\cite{KY}).
The following lemma is then immediate.
\begin{lemma}
\label{lem:randomness_of_function}
Let $d$ be an integer. If there exists an input distribution $\eta$ such that $\calR_\eta(f) > d$, then $\calR(f) > d$.
\end{lemma}
This means that in order to give a  lower bound on the randomness complexity of a function $f$, we can
 find an input distribution $\eta$ such that the randomness complexity of a function $f$ on $\eta$ is high.
Since we are  interested here in characterizing the randomness used in private protocols, in the rest 
of this section, when we use information terms such as $\SIC$,   we
 will make the randomness appear explicitly in the conditioning.

To make private protocols formally fit into our model (Section~\ref{subsec model}), we further technically modify them
such that  whenever a player does not send a message, it sends instead a special message indicating ``empty message''.
Such protocols formally fit in our model  and satisfy several additional properties. We call 
such protocols {\em proper synchronous} protocols
as defined below.
\begin{definition}
We say that a protocol as defined in Section~\ref{subsec model} is \textit{proper synchronous} if  there is an integer $t_f$ such that 
for every player  $i$, every input $x$, and every random assignment $r$ it holds that
\begin{itemize}
\item  In every (local) round $t<t_f$ player $i$ sends messages to all other players, and reads messages from all other players.
\item Player $i$ stops  at (local) round $t_f$.
\end{itemize}
\end{definition}

Observe that the above transformation from a synchronous protocol to a proper synchronous protocol preserves 
privacy (if the original protocol
was private), and the number of random bits used by the protocol does not change. We therefore prove
below our lower bound  for (private)  proper synchronous protocols.

\subsection{Public information cost}

The notion of public information cost was introduced in \cite{KRU}.

\begin{definition}
For any $k$-player protocol $\pi$ and any input distribution $\eta$, we define the \textit{public information cost} of $\pi$:
$$\PIC_\eta(\pi) = \sum\limits_{i=1}^k I(X_{-i} ; \Pi_i R_{-i} \mid X_i R_i R^p)~.$$
\end{definition}

Note that this definition of $\PIC_\eta(\pi)$ slightly differs from the one given in~\cite{KRU}, as the ``transcript''  $\Pi_i$ is
defined  in the present paper in a different way compared to the way it is defined in~\cite{KRU}.  However, since we work in this section in the 
setting of proper synchronous protocols, 
the two definitions of a ``transcript''
 are completely equivalent in terms of information, and thus the definition of $\PIC_\eta(\pi)$ in 
the present paper is equivalent to the one of \cite{KRU}.\footnote{In fact, they would be equivalent even 
if we were not restricting ourselves to proper synchronous protocols. This is because of the appearance of $X_i,R_i$ in the conditioning.}
\begin{definition}
For any function $f$ and any input distribution $\eta$, the zero-error public information cost of $f$ is
$$\PIC_\eta(f) =  \inf\limits_{\pi} \PIC_\eta(\pi)$$
where the infimum is taken over all protocols $\pi$ which compute  $f$ with $0$ error.
\end{definition}
It was shown~\cite{KRU}  that the public information cost can be used to prove 
randomness complexity lower bounds via the following theorem.

\begin{theorem}[\cite{KRU}]\label{thm:boundprivacy}
For any function of $k$ variables $f$, for any input distribution $\eta$, 
$$\calR_\eta(f) \ge \frac{\PIC_\eta(f) - k\cdot H_\eta(f(X))}{k}~.$$
\end{theorem}

We will need the following property of the public information cost.

\begin{theorem}[\cite{KRU}]\label{thm:minonpub}
For any function $f$ and input distribution $\eta$,
$$\PIC_\eta(f) =  \inf\limits_{\pi\text{ computing }f\text{, using only public coins}} \PIC_\eta(\pi)~.$$
\end{theorem}

\subsection{Randomness complexity of Disjointness}

We will prove that the switched multi-party information cost  gives a lower bound on  the public information cost.
Let $\mu$ be the input distribution for the function $\AND_k$ defined in Section~\ref{sec:AND}.

\begin{theorem}\label{prop:PIC>=SIC}
For any public-coins proper synchronous $k$-player protocol $\pi$, where the players have $n$-bits inputs $X$ 
from $(X,M,Z)\sim\mu^n$, it holds that 
$$\PIC_{\mu^n}(\pi) \ge \frac{1}{2}\SIC_{\mu^n}(\pi)~.$$
\end{theorem}

We start with a number of notations.
Recall that we consider a {\em proper synchronous} protocol $\pi$. 
We denote by $(\M{i}{\ell}{s}\,)_{\ell\ge 0}$ the sequence of messages sent by player $i$ in the protocol $\pi$, 
ordered by local round  of player $i$, and within  each round ordered by the index of the recipient player.
Further denote by $j(i,\ell)$ the player  to which message $\M{i}{l}{s}$ is sent.
Similarly,  we denote by $(\M{i}{\ell}{r}\,)_{\ell\ge 0}$ the sequence of messages received by player $i$, 
ordered in the same way.
Observe that since $\pi$ is proper synchronous there exists a function $\ell'(i,\ell)$, for $1 \leq i \leq k$, $\ell \geq 0$,
such that $\M{i}{\ell}{s}$ and  $\M{j(i,\ell)}{\ell'(i,\ell)}{r}$ denote the same message.\footnote{
$j(i,\ell) = (\ell \mod k) +1 $; $\ell'(i,\ell) = k\cdot \lfloor \frac{\ell+1}{ k} \rfloor  + (\ell \mod k)$.}
Further, for any $\ell_0 \geq 0$, let $\T{i}{\ell_0}{s}$
be the random variable representing the so-far history, i.e., all the messages sent by player $i$ and all the
 messages received by player $i$ until player $i$ sends message $\M{i}{\ell_0}{s}$ (for the same local round we define the order by the
  identity of the player sending or receiving the message).
In a similar way, define $\T{i}{\ell_0}{r} $
to be the random variable representing the messages sent by player $i$ and the messages received by
 player $i$ until it receives message $\M{i}{\ell_0}{r}$.
Further, we denote by $\overleftarrow{_{~}\Pi_i}$ the 
partial 
transcript 
of player $i$ composed only of the incoming messages. I.e., 
$\overleftarrow{_{~}\Pi_i}$ is the $(k-1)$-tuple $(\Pi_{i,j}^r)_{j\neq i}$.

Before giving the actual proof of Theorem~\ref{prop:PIC>=SIC} we
 define  two information theoretic measures, which we will use as intermediate quantities 
in that proof. These measures are defined only with respect to
the input distribution $\mu^n$, and thus we do not indicate the distribution in the notation of these measures.

\begin{definition}
$$\widehat{\IC}(\pi)=\sum\limits_{j=1}^k I(X_{-j} ; \overleftarrow{_{~}\Pi_j} \mid X_j R^p M Z)~.$$
\end{definition}
\begin{definition}
$$\widetilde{\IC}(\pi)=\sum\limits_{i=1}^k I(X_{-i} ; \Pi_i \mid X_i R^p Z)~.$$
\end{definition}

We now start the proof with two lemmas that relate  the intermediate measures that we just defined to the measure
$\PIC$. 
\begin{lemma}\label{lem:PIC>=PIChat}
For any public-coins protocol $\pi$,
$\widehat{\IC}(\pi) \le \PIC_{\mu^n}(\pi)$. 

\end{lemma}
\begin{proof}
For any $i\in\Ni{1}{k}$,
\begin{align*}
I(X_{-i} ; \overleftarrow{_{~}\Pi_i} \mid X_i R^p M Z) &\le I(X_{-i} ; \Pi_i \mid X_i R^p M Z)\\
&\le H(\Pi_i \mid X_i R^p M Z)\\
&\le H(\Pi_i \mid X_i R^p)\text{~~~(by Proposition \ref{prop:Hcondi})}\\
&=  H(\Pi_i \mid X_i R^p) - H(\Pi_i \mid X_i X_{-i} R^p) \text{~~~(because $X_i X_{-i} R^p$ determines $ \Pi_i$)}\\
&= I(X_{-i} ; \Pi_i \mid X_i R^p)~.
\end{align*}
Summing over $i\in\Ni{1}{k}$ concludes the proof.
\end{proof}

\begin{lemma}\label{lem:PIC>=PICtilde}
For any public-coins protocol $\pi$,
$\widetilde{\IC}(\pi) \le \PIC_{\mu^n}(\pi)$.
\end{lemma}
\begin{proof}
The proof is similar to the one of Lemma \ref{lem:PIC>=PIChat}.
For any $i\in\Ni{1}{k}$,
\begin{align*}
I(X_{-i} ; \Pi_i \mid X_i R^p Z) & \leq H(\Pi_i \mid X_i R^p Z)\\
&\le H(\Pi_i \mid X_i R^p)\text{~~~(by Proposition \ref{prop:Hcondi})}\\
&=  H(\Pi_i \mid X_i R^p) - H(\Pi_i \mid X_i X_{-i} R^p) \text{~~~(because $X_i X_{-i} R^p$ determines $ \Pi_i$)}\\
&= I(X_{-i} ; \Pi_i \mid X_i R^p)~.
\end{align*}
Summing over $i\in\Ni{1}{k}$ concludes the proof.
\end{proof}

The next two lemmas together relate $\SIC$ to the intermediate measures that we defined.
\begin{lemma}
\label{le:ic_tilde} 
For any public-coins   protocol $\pi$,
$\sum\limits_{i=1}^k I(M;\Pi_i \mid X_i R^p Z) \le \widetilde{\IC}(\pi)$.
\end{lemma}
\begin{proof}
We  prove that
$\forall~i\in\Ni{1}{k},~~ I(M;\Pi_i \mid X_i R^p Z) \le I(X_{-i} ; \Pi_i \mid X_i R^p Z)$.

\begin{align*}
I(M;\Pi_i \mid X_i R^p Z) &\le I(M X_{-i};\Pi_i \mid X_i R^p Z)\\
&\le I(X_{-i};\Pi_i \mid X_i R^p Z) + I(M;\Pi_i \mid X R^p Z) \text{~~~(chain rule)}\\
&\le I(X_{-i};\Pi_i \mid X_i R^p Z) + H(\Pi_i \mid X R^p Z) \\
&= I(X_{-i};\Pi_i \mid X_i R^p Z) \text{~~~(because $X R^p$ determines $\Pi_i$)}~.\\
\end{align*}

Summing over $i$ concludes the proof.
\end{proof}

The ideas behind the proof of the next lemma are similar to the ones developed in the proof of the lower 
bound on the randomness complexity of the Parity function in \cite{KRU}. However, the distribution and
the quantities involved being different, a 
different 
analysis is required here. We differ the proof of the next lemma to the appendix.

\begin{lemma}
\label{le:ic_hat}
For any public-coins proper synchronous protocol $\pi$, $\sum\limits_{i=1}^k I(X_i ; \Pi_i \mid R^p M Z) \le \widehat{\IC}(\pi)$.
\end{lemma}

We can now give the actual proof of Theorem~\ref{prop:PIC>=SIC}.
\begin{proof}[Proof of Theorem~\ref{prop:PIC>=SIC}]
By Lemma~\ref{le:ic_tilde} and Lemma~\ref{le:ic_hat} we have that
$$\SIC_{\mu^n}(\pi) \le \widetilde{\IC}(\pi) + \widehat{\IC}(\pi)~,$$
and using Lemma~\ref{lem:PIC>=PICtilde} and Lemma~\ref{lem:PIC>=PIChat} we get
$$\SIC_{\mu^n}(\pi) \le 2\cdot\PIC_{\mu^n}(\pi)~.$$
\end{proof}

We can now give a lower bound on the public information cost of the disjointness function.

\begin{theorem}\label{thm:PICofDISJ}
Let $k>3$. For any proper synchronous protocol $\pi$ computing $\Disj_k^n$ it holds that 
$$\PIC_{\mu^n}(\pi) = \Omega(kn)~.$$
\end{theorem}
 
 \begin{proof}\begin{sloppypar}
By Theorem \ref{thm:minonpub}, we only have to consider public-coins protocols. Observe that
by adding an additional round to all players, such that, say, player $1$  sends to player $2$ his output, and
 all other  ${k(k-1)-1}$ messages are constant, we can convert $\pi$ into a protocol $\pi'$ {\em externally} computing $\Disj_k^n$.
By Theorem \ref{thm:SIC>=kn} and Theorem~\ref{prop:PIC>=SIC}, it holds  that
${\PIC_{\mu^n}(\pi')=\Omega(kn)}$. 
Since  $\PIC_{\mu^n}(\pi')\le\PIC_{\mu^n}(\pi)+1$, we get that ${\PIC_{\mu^n}(\pi)=\Omega(kn)}$.\end{sloppypar}
\end{proof}

Our lower bound on the randomness complexity of the disjointness function then follows.

\begin{theorem}\label{thm:obli:Rdisj>kn}
Let $k>3$.
$\calR(\Disj_k^n) = \Omega(n)$.
\end{theorem}
\begin{proof}
By Theorem \ref{thm:PICofDISJ}, $\PIC_{\mu^n}(\Disj_k^n)=\Omega(kn)$.
But $H_{\mu^n}(\Disj_k^n)=0$, and applying Theorem~\ref{thm:boundprivacy} 
we get 
$$ \calR_{\mu^n}(\Disj_k^n)\ge \frac{\Omega(kn)}{k} = \Omega(n)~.$$
By Lemma~\ref{lem:randomness_of_function} we have that any proper synchronous
protocol privately computing  $\Disj_k^n$ must use $\Omega(n)$ random bits.
Since by Lemma~\ref{lemma:alaKonig} any (synchronous) private protocol can be transformed into a proper synchronous
one without increasing the number of random bits used we have that
$$\calR(\Disj_k^n)=\Omega(n)~.$$
\end{proof}
\LongVersionEnd

\section{Conclusions and open problems}
\label{sec:conclusions}

We introduce  new  models and new information theoretic tools for the study of  
communication complexity,  and other complexity measures, in the natural  peer-to-peer, multi-party, 
number-in-hand setting. We prove a number of properties of
our new models and measures, and exemplify their effectiveness by proving two  
lower bounds on communication complexity, as well as a lower bound on the
amount of randomness necessary for certain private computations.

 To the best of our knowledge, our lower bounds  on communication complexity are the first tight 
(non-trivial)
lower bounds on communication complexity in the natural {\em peer-to-peer} multi-party setting, and 
our lower bound
on the randomness complexity of private computations is the first  that grows with the size of the input, while the 
computed function is a boolean one (i.e., the size of the output does not grow with the size 
of the input).

We believe that our models and tools may find additional applications and may open the way to further  study of the
natural peer-to-peer setting and to the building of a more solid bridge between the the fields
of communication complexity and of distributed computation.

Our work raises a number of
questions. 
First, how can one relax the  restrictions that we impose on the general asynchronous model  
and still prove communication complexity lower bounds in a peer-to-peer setting?
Our work seems to suggest that  novel techniques and ideas, possibly not based on information theory,
  are necessary for this task, and it would be 
most interesting to find those.
Second, it would be interesting to identify 
the necessary and sufficient conditions that guarantee the ``rectangularity'' property of communication 
protocols in a peer-to-peer setting. While this property is fundamental to the analysis of two-party protocols, it
 turns out that once one turns to the multi-party peer-to-peer setting, not only does  this property become
 subtle to prove,  but also this property does not always hold. Given the central (and sometimes implicit) role 
of the rectangularity property in the literature, it would be  interesting to identify when it holds in the multi-party 
peer-to-peer number-in-hand setting.

\paragraphA{Acknowledgments.}
We thank Iordanis Kerenidis and Rotem Oshman  for very useful discussions.

{
\nocite{CT}
\nocite{Kleene}
\nocite{KY}
\small
\bibliography{bib_dist}

\begin{thebibliography}{10}

\bibitem{AsharovL17}
Gilad Asharov and Yehuda Lindell.
\newblock A full proof of the {BGW} protocol for perfectly secure multiparty
  computation.
\newblock {\em J. Cryptology}, 30(1):58--151, 2017.

\bibitem{Bar-YehudaCKO93}
Reuven Bar{-}Yehuda, Benny Chor, Eyal Kushilevitz, and Alon Orlitsky.
\newblock Privacy, additional information and communication.
\newblock {\em {IEEE} Transactions on Information Theory}, 39(6):1930--1943,
  1993.

\bibitem{BJKS}
Ziv Bar{-}Yossef, T.~S. Jayram, Ravi Kumar, and D.~Sivakumar.
\newblock An information statistics approach to data stream and communication
  complexity.
\newblock {\em J. Comput. Syst. Sci.}, 68(4):702--732, 2004.

\bibitem{BBCR}
Boaz Barak, Mark Braverman, Xi~Chen, and Anup Rao.
\newblock How to compress interactive communication.
\newblock In {\em Proceedings of the 42nd ACM symposium on Theory of
  computing}, STOC '10, pages 67--76, New York, NY, USA, 2010. ACM.

\bibitem{BGW}
Michael Ben-Or, Shafi Goldwasser, and Avi Wigderson.
\newblock Completeness theorems for non-cryptographic fault-tolerant
  distributed computation.
\newblock In {\em Proceedings of the twentieth annual ACM symposium on Theory
  of computing}, STOC '88, pages 1--10, New York, NY, USA, 1988. ACM.

\bibitem{BSPV}
C.~Blundo, A.~De~Santis, G.~Persiano, and U.~Vaccaro.
\newblock Randomness complexity of private computation.
\newblock {\em computational complexity}, 8(2):145--168, 1999.

\bibitem{Bra}
Mark Braverman.
\newblock Interactive information complexity.
\newblock {\em {SIAM} J. Comput.}, 44(6):1698--1739, 2015.

\bibitem{BEOPV}
Mark Braverman, Faith Ellen, Rotem Oshman, Toniann Pitassi, and Vinod
  Vaikuntanathan.
\newblock A tight bound for set disjointness in the message-passing model.
\newblock In {\em 54th Annual {IEEE} Symposium on Foundations of Computer
  Science, {FOCS} 2013, 26-29 October, 2013, Berkeley, CA, {USA}}, pages
  668--677. {IEEE} Computer Society, 2013.

\bibitem{BGPW}
Mark Braverman, Ankit Garg, Denis Pankratov, and Omri Weinstein.
\newblock From information to exact communication.
\newblock In {\em Proceedings of the 45th annual ACM symposium on Symposium on
  theory of computing}, STOC '13, pages 151--160, New York, NY, USA, 2013. ACM.

\bibitem{BO}
Mark Braverman and Rotem Oshman.
\newblock On information complexity in the broadcast model.
\newblock In Chryssis Georgiou and Paul~G. Spirakis, editors, {\em Proceedings
  of the 2015 {ACM} Symposium on Principles of Distributed Computing, {PODC}
  2015, Donostia-San Sebasti{\'{a}}n, Spain, July 21 - 23, 2015}, pages
  355--364. {ACM}, 2015.

\bibitem{BR}
Mark Braverman and Anup Rao.
\newblock Information equals amortized communication.
\newblock {\em {IEEE} Trans. Information Theory}, 60(10):6058--6069, 2014.

\bibitem{CK}
Amit Chakrabarti and Sagar Kale.
\newblock Strong fooling sets for multi-player communication with applications
  to deterministic estimation of stream statistics.
\newblock In Irit Dinur, editor, {\em {IEEE} 57th Annual Symposium on
  Foundations of Computer Science, {FOCS} 2016, 9-11 October 2016, New
  Brunswick, New Jersey, {USA}}, pages 41--50. {IEEE} Computer Society, 2016.

\bibitem{CKS}
Amit Chakrabarti, Subhash Khot, and Xiaodong Sun.
\newblock Near-optimal lower bounds on the multi-party communication complexity
  of set disjointness.
\newblock In {\em In IEEE Conference on Computational Complexity}, pages
  107--117, 2003.

\bibitem{CSWY}
Amit Chakrabarti, Yaoyun Shi, Anthony Wirth, and Andrew Chi-Chih Yao.
\newblock Informational complexity and the direct sum problem for simultaneous
  message complexity.
\newblock In {\em FOCS}, pages 270--278, 2001.

\bibitem{CM}
Arkadev Chattopadhyay and Sagnik Mukhopadhyay.
\newblock Tribes is hard in the message passing model.
\newblock In Mayr and Ollinger \cite{DBLP:conf/stacs/2015}, pages 224--237.

\bibitem{CP}
Arkadev Chattopadhyay and Toniann Pitassi.
\newblock The story of set disjointness.
\newblock {\em SIGACT News}, 41(3):59--85, September 2010.

\bibitem{CRR}
Arkadev Chattopadhyay, Jaikumar Radhakrishnan, and Atri Rudra.
\newblock Topology matters in communication.
\newblock In {\em 55th {IEEE} Annual Symposium on Foundations of Computer
  Science, {FOCS} 2014, Philadelphia, PA, USA, October 18-21, 2014}, pages
  631--640, 2014.

\bibitem{CR}
Arkadev Chattopadhyay and Atri Rudra.
\newblock The range of topological effects on communication.
\newblock In Magn{\'{u}}s~M. Halld{\'{o}}rsson, Kazuo Iwama, Naoki Kobayashi,
  and Bettina Speckmann, editors, {\em Automata, Languages, and Programming -
  42nd International Colloquium, {ICALP} 2015, Kyoto, Japan, July 6-10, 2015,
  Proceedings, Part {II}}, volume 9135 of {\em Lecture Notes in Computer
  Science}, pages 540--551. Springer, 2015.

\bibitem{CCD}
David Chaum, Claude Cr{\'e}peau, and Ivan Damgard.
\newblock Multiparty unconditionally secure protocols.
\newblock In {\em Proceedings of the twentieth annual ACM symposium on Theory
  of computing}, STOC '88, pages 11--19, New York, NY, USA, 1988. ACM.

\bibitem{CT}
Thomas~M. Cover and Joy~A. Thomas.
\newblock {\em Elements of Information Theory (Wiley Series in
  Telecommunications and Signal Processing)}.
\newblock Wiley-Interscience, 2006.

\bibitem{DF}
Danny Dolev and Tom{\'{a}}s Feder.
\newblock Multiparty communication complexity.
\newblock In {\em 30th Annual Symposium on Foundations of Computer Science,
  Research Triangle Park, North Carolina, USA, 30 October - 1 November 1989},
  pages 428--433. {IEEE} Computer Society, 1989.

\bibitem{EOPV}
Faith Ellen, Rotem Oshman, Toniann Pitassi, and Vinod Vaikuntanathan.
\newblock {Brief Announcement: Private Channel Models in Multi-party
  Communication Complexity}.
\newblock In {\em {27th International Symposium on Distributed Computing
  (DISC), Jerusalem, Israel}}, pages 575--576, 2013.

\bibitem{FKN}
Uri Feige, Joe Killian, and Moni Naor.
\newblock A minimal model for secure computation (extended abstract).
\newblock In {\em Proceedings of the Twenty-sixth Annual ACM Symposium on
  Theory of Computing}, STOC '94, pages 554--563, New York, NY, USA, 1994. ACM.

\bibitem{FHW}
Silvio Frischknecht, Stephan Holzer, and Roger Wattenhofer.
\newblock Networks cannot compute their diameter in sublinear time.
\newblock In Yuval Rabani, editor, {\em Proceedings of the Twenty-Third Annual
  {ACM-SIAM} Symposium on Discrete Algorithms, {SODA} 2012, Kyoto, Japan,
  January 17-19, 2012}, pages 1150--1162. {SIAM}, 2012.

\bibitem{GG}
Anna G{\'{a}}l and Parikshit Gopalan.
\newblock Lower bounds on streaming algorithms for approximating the length of
  the longest increasing subsequence.
\newblock {\em {SIAM} J. Comput.}, 39(8):3463--3479, 2010.

\bibitem{GR}
Anna G{\'{a}}l and Adi Ros{\'{e}}n.
\newblock Omega(log n) lower bounds on the amount of randomness in 2-private
  computation.
\newblock {\em {SIAM} J. Comput.}, 34(4):946--959, 2005.

\bibitem{Gro}
Andre Gronemeier.
\newblock Asymptotically optimal lower bounds on the nih-multi-party
  information complexity of the and-function and disjointness.
\newblock In Susanne Albers and Jean{-}Yves Marion, editors, {\em 26th
  International Symposium on Theoretical Aspects of Computer Science, {STACS}
  2009, February 26-28, 2009, Freiburg, Germany, Proceedings}, volume~3 of {\em
  LIPIcs}, pages 505--516. Schloss Dagstuhl - Leibniz-Zentrum fuer Informatik,
  Germany, 2009.

\bibitem{HRVZ}
Zengfeng Huang, Bozidar Radunovic, Milan Vojnovic, and Qin Zhang.
\newblock Communication complexity of approximate matching in distributed
  graphs.
\newblock In Mayr and Ollinger \cite{DBLP:conf/stacs/2015}, pages 460--473.

\bibitem{Jay}
T.~S. Jayram.
\newblock Hellinger strikes back: A note on the multi-party information
  complexity of and.
\newblock In {\em Proceedings of the 12th International Workshop and 13th
  International Workshop on Approximation, Randomization, and Combinatorial
  Optimization. Algorithms and Techniques}, APPROX '09 / RANDOM '09, pages
  562--573, Berlin, Heidelberg, 2009. Springer-Verlag.

\bibitem{KS}
Bala Kalyanasundaram and Georg Schintger.
\newblock The probabilistic communication complexity of set intersection.
\newblock {\em SIAM J. Discret. Math.}, 5(4):545--557, November 1992.

\bibitem{KRU}
Iordanis Kerenidis, Adi Ros{\'{e}}n, and Florent Urrutia.
\newblock Multi-party protocols, information complexity and privacy.
\newblock In Piotr Faliszewski, Anca Muscholl, and Rolf Niedermeier, editors,
  {\em 41st International Symposium on Mathematical Foundations of Computer
  Science, {MFCS} 2016, August 22-26, 2016 - Krak{\'{o}}w, Poland}, volume~58
  of {\em LIPIcs}, pages 57:1--57:16. Schloss Dagstuhl - Leibniz-Zentrum fuer
  Informatik, 2016.

\bibitem{Kleene}
Stephen~Cole Kleene.
\newblock {\em Mathematical Logic}.
\newblock Dover, 2002.
\newblock Reprint of the John Wiley \& Sons, Inc., New York, 1967 edition.

\bibitem{KY}
D.~Knuth and A.~Yao.
\newblock {\em Algorithms and Complexity: New Directions and Recent Results},
  chapter The complexity of nonuniform random number generation.
\newblock Academic Press, 1976.

\bibitem{KorhonenS17}
Janne~H. Korhonen and Jukka Suomela.
\newblock Brief announcement: Towards a complexity theory for the congested
  clique.
\newblock In Andr{\'{e}}a~W. Richa, editor, {\em 31st International Symposium
  on Distributed Computing, {DISC} 2017, October 16-20, 2017, Vienna, Austria},
  volume~91 of {\em LIPIcs}, pages 55:1--55:3. Schloss Dagstuhl -
  Leibniz-Zentrum fuer Informatik, 2017.

\bibitem{KorhonenS17_arxiv}
Janne~H. Korhonen and Jukka Suomela.
\newblock Towards a complexity theory for the congested clique.
\newblock {\em CoRR}, abs/1705.03284, 2017.

\bibitem{KM}
Eyal Kushilevitz and Yishay Mansour.
\newblock Randomness in private computations.
\newblock {\em {SIAM} J. Discrete Math.}, 10(4):647--661, 1997.

\bibitem{KN}
Eyal Kushilevitz and Noam Nisan.
\newblock {\em Communication complexity}.
\newblock Cambridge University Press, 1997.

\bibitem{KOR}
Eyal Kushilevitz, Rafail Ostrovsky, and Adi Ros{\'{e}}n.
\newblock Characterizing linear size circuits in terms of pricacy.
\newblock {\em J. Comput. Syst. Sci.}, 58(1):129--136, 1999.

\bibitem{DBLP:conf/stacs/2015}
Ernst~W. Mayr and Nicolas Ollinger, editors.
\newblock {\em 32nd International Symposium on Theoretical Aspects of Computer
  Science, {STACS} 2015, March 4-7, 2015, Garching, Germany}, volume~30 of {\em
  LIPIcs}. Schloss Dagstuhl - Leibniz-Zentrum fuer Informatik, 2015.

\bibitem{MNSW}
Peter~Bro Miltersen, Noam Nisan, Shmuel Safra, and Avi Wigderson.
\newblock On data structures and asymmetric communication complexity.
\newblock {\em J. Comput. Syst. Sci.}, 57(1):37--49, 1998.

\bibitem{PVZ}
Jeff~M. Phillips, Elad Verbin, and Qin Zhang.
\newblock Lower bounds for number-in-hand multiparty communication complexity,
  made easy.
\newblock {\em {SIAM} J. Comput.}, 45(1):174--196, 2016.

\bibitem{Raz}
A.~A. Razborov.
\newblock On the distributional complexity of disjointness.
\newblock {\em Theor. Comput. Sci.}, 106(2):385--390, December 1992.

\bibitem{RosenU19}
Adi Ros{\'{e}}n and Florent Urrutia.
\newblock A new approach to multi-party peer-to-peer communication complexity.
\newblock In Avrim Blum, editor, {\em 10th Innovations in Theoretical Computer
  Science Conference, {ITCS} 2019, January 10-12, 2019, San Diego, California,
  {USA}}, volume 124 of {\em LIPIcs}, pages 64:1--64:19. Schloss Dagstuhl -
  Leibniz-Zentrum fuer Informatik, 2019.

\bibitem{SHKKNPPW}
Atish~Das Sarma, Stephan Holzer, Liah Kor, Amos Korman, Danupon Nanongkai,
  Gopal Pandurangan, David Peleg, and Roger Wattenhofer.
\newblock Distributed verification and hardness of distributed approximation.
\newblock {\em CoRR}, abs/1011.3049, 2010.

\bibitem{Sha}
C.~E. Shannon.
\newblock {A mathematical theory of communication}.
\newblock {\em Bell system technical journal}, 27, 1948.

\bibitem{WZ}
David~P. Woodruff and Qin Zhang.
\newblock Tight bounds for distributed functional monitoring.
\newblock {\em CoRR}, abs/1112.5153, 2011.

\bibitem{WZ2}
David~P. Woodruff and Qin Zhang.
\newblock When distributed computation does not help.
\newblock {\em CoRR}, abs/1304.4636, 2013.

\bibitem{WZ3}
David~P. Woodruff and Qin Zhang.
\newblock An optimal lower bound for distinct elements in the message passing
  model.
\newblock In {\em Proceedings of the Twenty-fifth Annual ACM-SIAM Symposium on
  Discrete Algorithms}, SODA '14, pages 718--733, Philadelphia, PA, USA, 2014.
  Society for Industrial and Applied Mathematics.

\bibitem{Yao82}
Andrew~Chi{-}Chih Yao.
\newblock Protocols for secure computations (extended abstract).
\newblock In {\em 23rd Annual Symposium on Foundations of Computer Science,
  Chicago, Illinois, USA, 3-5 November 1982}, pages 160--164. {IEEE} Computer
  Society, 1982.

\end{thebibliography}
}

\LongVersion
\appendix

\section{Background in Information Theory}
We give a reminder on basic information theory tools that are of use in the present paper.
 A good reference is the book of Cover and Thomas \cite{CT}.
We always consider a probability space over a discrete domain.

\subsection{Entropy and mutual information}

\begin{definition}
The \textit{entropy}\footnote{In this paper we refer to the binary entropy by simply saying ``entropy''.} of a random variable $X$ is
$$H(X) = \sum\limits_{x}\Pr[X=x]\log\left(\frac{1}{\Pr[X=x]}\right)~.$$
We further use the notation 
$$ H(X \mid Y=y) = \sum\limits_{x}\Pr[X=x  \mid Y=y ]\log\left(\frac{1}{\Pr[X=x  \mid Y=y]}\right)~.$$
The \textit{conditional entropy} $H(X \mid Y)$ is defined as $\Exp\limits_y[H(X \mid Y=y)]$.
\end{definition}

\begin{proposition}\label{prop:Hcondi}
For any random variables $X$ and $Y$, $H(X \mid Y) \le H(X)$.
\end{proposition}

The entropy of a random variable is always non-negative. 

\begin{theorem}[Shannon]\label{thm:Shannon}
For all prefix-free finite set $\calX \subseteq \{0,1\}^{*}$ and all random variable $X$ with support $\text{supp}(X) \subseteq \calX$, it holds
$$H(X) \le \Exp[ | X |]~.$$
\end{theorem}

\begin{definition}
The \textit{mutual information between two random variables $X,Y$} is
$$I(X;Y) = H(X) - H(X \mid Y)~.$$
The mutual information of $X$ and $Y$ conditioned on $Z$ is 
$$I(X;Y \mid Z) = H(X \mid Z) - H(X \mid YZ)~.$$
\end{definition}
The mutual information measures the change in the entropy of $X$ when one learns the value of $Y$. It is symmetric, and non-negative.

\begin{proposition}\label{prop:indvarandI}
For any random variables $X$, $Y$ and $Z$, $I(X;Y \mid Z) = 0$ if and only if $X$ and $Y$ are independent conditioned on every possible value of $Z$.
\end{proposition}

We will use extensively the following proposition, known under the name of \textit{chain rule}.
\begin{proposition}\label{prop:CR}
For any random variables $A$, $B$, $C$, $D$,
$$I(AB ; C \mid D) = I(A ; C \mid D) + I(B ; C \mid DA)~.$$
\end{proposition}

The \textit{data processing inequality} expresses the fact that information can only be lost when applying a function to a random variable.

\begin{proposition}\label{prop:dataproc}
For any random variables $X$, $Y$, $Z$, and any function $f$
$$I(X ; f(Y) \mid Z)\le I(X ; Y \mid Z)~.$$
\end{proposition}

We will occasionally make use of the two following lemmas, which allow to add or remove a random variable from the conditioning.

\begin{lemma}[\cite{Bra}]\label{lem:Bra2.9}
For any random variables $A$, $B$, $C$, $D$ such that $I(B ; D \mid AC) = 0$,
$$I(A;B \mid C) \ge I(A ; B \mid CD)~.$$
\end{lemma}
\begin{lemma}[\cite{Bra}]\label{lem:Bra2.10}
For any random variables $A$, $B$, $C$, $D$ such that $I(B ; D \mid C) = 0$,
$$I(A;B \mid C) \le I(A ; B \mid CD)~.$$
\end{lemma}

We  further give  a  lemma which is an certain extension of the data processing inequality, 
allowing the processing to depend also on part of the conditioning.

\begin{lemma}\label{lem:tec1}
Let $A$, $B$, $C$, $D$, $\phi=\varphi(C,B)$ be random variables.
Then,
$$
I(A;\phi\mid C D) \leq I(A;B\mid C D)~.
$$
\end{lemma}
\begin{proof}
\begin{align*}
I(A;\phi\mid C D)&=I(A;\varphi(C,B)\mid C D)\\
&=\Exp\limits_c[I(A;\varphi(c,B)\mid C=c, D)]\\
&\le\Exp\limits_c[I(A;B\mid C=c, D)]~~~\text{ (by the data processing inequality, Proposition~\ref{prop:dataproc})}\\
&\le I(A;B\mid C D).
\end{align*}
\end{proof}

\subsection{Hellinger distance}

We will make an extensive use of the Hellinger distance.

\begin{definition}
Let $P$ and $Q$ be two distributions over a domain $\Omega$. The \textit{Hellinger distance} between $P$ and $Q$ is $h(P,Q)=\frac{1}{\sqrt{2}}\sqrt{\sum\limits_{\omega\in\Omega}\mid\sqrt{P(\omega)}-\sqrt{Q(\omega)}\mid^2}$.
\end{definition}

It can be easily checked that the Hellinger distance is indeed a ``distance''.  
When using the square of the Hellinger distance, we often use the following identity.

\begin{proposition}\label{prop:Hellsquare}
Let $P$ and $Q$ be two distributions over a domain $\Omega$.
$$h^2(P,Q)=1-\sum\limits_{\omega\in\Omega}\sqrt{P(\omega)Q(\omega)}~.$$
\end{proposition}

Hellinger distance can be related to mutual information by the following relation.
\begin{lemma}[\cite{BJKS}]\label{lem:ICandHel}
Let $\eta_0,\eta_1$ be two distributions over the same domain, and suppose that $Y$ is generated as follows: 
first select $S$ uniformly in $\{0,1\}$, and then sample $Y$ according to $\eta_S$. Then $I(S;Y) \ge [h(\eta_0,\eta_1)]^2$.
\end{lemma}

Another useful measure is the statistical distance.
\begin{definition}
Let $P$ and $Q$ be two distributions over a domain $\Omega$. The \textit{statistical distance} between $P$ and $Q$ is 
$$\Delta(P,Q)=\max\limits_{\Omega'\subseteq\Omega}\mid P(\Omega') - Q(\Omega')\mid~.$$
\end{definition}

Hellinger distance and statistical distance are related by the following relation.
\begin{lemma}\label{lem:HelAndStat}
Let $P$ and $Q$ be two distributions over the same domain. $h(P,Q)\ge\frac{1}{\sqrt{2}}\Delta(P,Q)$.
\end{lemma}

\section{A technical lemma}

\begin{claim}
\label{cl:entropy_of_calculation}
Let $\pi$ be a  protocol,  let $i$ be a given player, and let $0 \leq \epsilon  \leq \frac{1}{2}$ be fixed.
 If, when running $\pi$, player $i$ $\epsilon$-computes a boolean function $f$, 
then $H(f(X) \mid X_i R_i R^p \Pi_i) \le h(\epsilon)$, where $h$ is the binary entropy function.
\end{claim}
\begin{proof}
Let $\theta$ be the (deterministic) function that takes as parameter $(x_i, r_i, r^p, \pi_i)$ and returns the output of player $i$.
Define the random variable $P=\theta(X_i, R^p, R_i,\Pi_i)$, and  
the random variable $M=1 - \delta_{f(X),P}$, i.e., the indicator variable of the event $f(X) \neq P$. Observe that

\begin{align*}
\Pr(M=1) &= \Exp[M]\\
&= \sum\limits_x \Pr(X=x) \Exp[M \mid X=x]\\
&=\sum\limits_x \Pr(X=x) \Pr(M=1 \mid X=x)\\
&\le\sum\limits_x \Pr(X=x) \cdot  \epsilon \text{~~~(since player $i$ $\epsilon$-computes $f$)}\\
&\le \epsilon~.
\end{align*}

Thus we have
\begin{align*}
H(f(X) \mid X_i R_i R^p \Pi_i) &\le H(f(X) \mid P)\text{~~(data processing inequality)}\\
&= H(M \mid P) \text{~~~(since, given $P$, there is a bijection between $f(X)$ and $M$)}\\
&\le H(M)\\
&= h(\Pr(M=1)) \text{~~~(since M is binary)}\\
&\le h(\epsilon) \text{~~~(since $h$ is increasing in $[0,1/2]$)}.
\end{align*}
\end{proof}

\section{Some of the proofs}

This section contains the proofs  that were deferred to the appendix.

\begin{proof}[Proof of Lemma~\ref{lem:rectang}]
We prove the claim for an arbitrary player $i\in\Ni{1}{k}$.
To prove the statement of the lemma  define, 
for $x'_i\in\calX_i$, $$q_i(x'_i,\overline{\tau})=\Pr[(x'_i,R_i)\in\calI_i(\overline{\tau})]~,$$
and for $x'_{-i}\in\calX_{-i}$, $$q_{-i}(x'_{-i},\overline{\tau})=\Pr[(x'_{-i},R_{-i})\in\calJ_i(\overline{\tau})]~.$$
We have, for $x\in\calX$, 
\begin{align*}
\Pr[\Pi_i(x)=\overline{\tau}] &= \Pr[(x,R)\in\calA_i(\overline{\tau})]\\
&=\Pr[(x_i,R_i)\in\calI_i(\overline{\tau})~\&~(x_{-i},R_{-i})\in\calJ_i(\overline{\tau})]
~~~~~\mbox{by Lemma~\ref{le:det_rectangularity}}\\
&=\Pr[(x_i,R_i)\in\calI_i(\overline{\tau})] \times \Pr[(x_{-i},R_{-i})\in\calJ_i(\overline{\tau})]\\
&=q_i(x_i,\overline{\tau})q_{-i}(x_{-i},\overline{\tau})~.
\end{align*}

We now prove the second claim.
Define, for $x'_{-i},\in\calX_{-i}$, $$p_{-i}(x'_{-i},\tau)=\Pr[(x'_{-i},R_{-i})\in\calH_i(\tau)]~.$$

We have
\begin{align*}
\Pr[\Pi(x)=\tau]&=\Pr[(x,R)\in\calB(\tau))]\\
&=\Pr[(x_i,R_i)\in\calI_i(\tau_i) ~\&~ (x_{-i},R_{-i})\in\calH_i(\tau))]
~~~~~\mbox{by Lemma~\ref{le:det_rectangularity}}\\
&=\Pr[(x_i,R_i)\in\calI_i(\tau_i)] \times \Pr[(x_{-i},R_{-i})\in\calH_i(\tau))]\\
&=q_i(x_i,\tau_i)p_{-i}(x_{-i},\tau)~.
\end{align*}
\end{proof}

\begin{proof}[Proof of Lemma~\ref{lem:HelErr}]\begin{sloppypar}
By Lemma~\ref{lem:HelAndStat}, we only need to show that ${\Delta(\Pi(x),\Pi(y))\ge1-2\epsilon}$. 
By definition, there exists a function $\theta$ taking as input the possible transcripts of ${\pi}$ and verifying ${\forall~x \in \calX, \Pr[\theta(\Pi(x)) = f(x)] \ge 1 - \epsilon}$.\\
Let ${\Omega'=\theta^{-1}(f(x))}$.
We  have
${\Pr[\Pi(x)\in\Omega'] = \Pr[\theta(\Pi(x))=f(x)] \ge 1 - \epsilon}$ and\end{sloppypar}
\begin{align*}
\Pr[\Pi(y)\in\Omega'] &= \Pr[\Pi(y)\in\theta^{-1}(f(x))]\\
&\le \Pr[\Pi(y)\not\in\theta^{-1}(f(y))]~~~~~\text{since $\theta^{-1}(f(x))\cap\theta^{-1}(f(y))=\varnothing$}\\
&\le 1-\Pr[\Pi(y)\in\theta^{-1}(f(y))]\\
&\le 1-(1-\epsilon) = \epsilon.
\end{align*}

Thus
\begin{align*}
\Delta(\Pi(x),\Pi(y)) &\ge \Pr[\Pi(x)\in\Omega']-\Pr[\Pi(y)\in\Omega']\\
&\ge (1-\epsilon) - \epsilon = 1-2\epsilon.
\end{align*}
\end{proof}

\begin{proof}[Proof of Lemma~\ref{lem:diag}]
Let ${\cal T}$ be the set of all possible transcripts of $\pi$. In what follows we simplify notation and
 write $\sum\limits_{\tau}$ instead of $\sum\limits_{\tau \in{\cal T}}$.
Using Proposition~\ref{prop:Hellsquare},
\begin{align*}
1-h^2(\Pi(x),\Pi(y))&=\sum\limits_{\tau}\sqrt{\Pr[\Pi(x)=\tau]\Pr[\Pi(y)=\tau]}\\
&=\sum\limits_{\tau}\sqrt{q_i(x_i,\tau_i)p_{-i}(x_{-i},\tau)q_i(y_i,\tau_i)p_{-i}(y_{-i},\tau)}\text{~~~~By  
Lemma~\ref{lem:rectang}}\\
&=\sum\limits_{\tau}\sqrt{q_i(x_i,\tau_i)q_i(y_i,\tau_i)}\sqrt{p_{-i}(x_{-i},\tau)p_{-i}(y_{-i},\tau)}\\
&\le\sum\limits_{\tau}\frac{q_i(x_i,\tau_i)+q_i(y_i,\tau_i)}{2}\sqrt{p_{-i}(x_{-i},\tau)p_{-i}(y_{-i},\tau)}\\
&\le\frac{1}{2}\left(\sum\limits_{\tau}\sqrt{q_i(x_i,\tau_i)p_{-i}(x_{-i},\tau)q_i(x_i,\tau_i)p_{-i}(y_{-i},\tau)}\right.\\
&\,+\left.\sum\limits_{\tau} \sqrt{q_i(y_i,\tau_i)p_{-i}(x_{-i},\tau)q_i(y_i,\tau_i)p_{-i}(y_{-i},\tau)}\right)\\
&\le\frac{1}{2}\left(\sum\limits_{\tau}\sqrt{\Pr[\Pi(x)=\tau]\Pr[\Pi(y_{[i\leftarrow x_i]})=\tau]}\right.\\
&~~~~~\,+\left.\sum\limits_{\tau} \sqrt{\Pr[\Pi(x_{[i\leftarrow y_i]})=\tau]\Pr[\Pi(y)=\tau]}\right)\\
&\le\frac{1}{2}\left[1-h^2(\Pi(x),\Pi(y_{[i\leftarrow x_i]}))+1-h^2(\Pi(x_{[i\leftarrow y_i]}),\Pi(y))\right]\\
&\le 1-\frac{1}{2}\left[h^2(\Pi(x),\Pi(y_{[i\leftarrow x_i]}))+h^2(\Pi(x_{[i\leftarrow y_i]}),\Pi(y))\right].
\end{align*}
\end{proof}

\begin{proof}[Proof of Lemma~\ref{lem:rectangspe}]
\adifuture{go over correctness once more}
We write $\Pr[\Pi_i=\tau_i\mid X_i=x', M=m, Z=z]$ as 
$$\sum\limits_{x\in\{0,1\}^k}(\Pr[X=x\mid X_i=x', M=m, Z=z]\times \Pr[\Pi_i=\tau_i\mid X=x, X_i=x', M=m, Z=z])~.$$

Note that $$\Pr[X=x\mid X_i=x', M=m, Z=z] = \delta_{x_i,x'}\Pr[X_{-i}=x_{-i}\mid M=m, Z=z]~,$$ 
since, conditioned on $M=m$ and $Z=z$,  $X_i$ and $X_{-i}$ are independent. Further note that for $x$ such that $x_i=x'$,
$$\Pr[\Pi_i=\tau_i\mid X=x, X_i=x', M=m, Z=z] = \Pr[\Pi_i(x)=\tau_i]~.$$

By Lemma~\ref{lem:rectang}, there exist functions $q_i$ and $q_{-i}$ such that 
$$\forall~x\in\{0,1\}^k,~\Pr[\Pi_i(x)=\tau_i]=q_i(x_i,\tau_i)q_{-i}(x_{-i},\tau_i)~.$$
Therefore we can write
\begin{align*}
\Pr[\Pi_i=\tau_i\mid X_i=x, M=m, Z=z] &= 
\sum\limits_{x\in\{0,1\}^k} (\delta_{x_i,x'}q_i(x_i,\tau_i)q_{-i}(x_{-i},\tau_i)\times
\Pr[X_{-i}=x_{-i}\mid M=m, Z=z])\\
&~~\\
&=q_i(x',\tau_i)\sum\limits_{\hat{x}\in\{0,1\}^{k-1}}(q_{-i}(\hat{x},\tau_i)\times
\Pr[X_{-i}=\hat{x}\mid M=m, Z=z])\\
&~~\\
&=q_i(x',\tau_i)c_i(m,z,\tau_i)~,
\end{align*}
where $c_i(m,z,\tau_i)=\sum\limits_{\hat{x}\in\{0,1\}^{k-1}}q_{-i}(\hat{x},\tau_i)\Pr[X_{-i}=\hat{x}\mid M=m, Z=z]$.\\

The proof of the second statement is similar:

\begin{align*}
&\Pr[\Pi=\tau\mid X_i=x, M=m, Z=z]=  \\
&\sum\limits_{x\in\{0,1\}^k} (
\Pr[X=x\mid X_i=x', M=m, Z=z]~\times~
\Pr[\Pi=\tau\mid X=x, X_i=x', M=m, Z=z]).
\end{align*}
Note that 
$\Pr[X=x\mid X_i=x', M=m, Z=z] = \delta_{x_i,x'}\Pr[X_{-i}=x_{-i}\mid M=m, Z=z]$, 
since, conditioned on $M=m$ and $Z=z$, $X_i$ and $X_{-i}$ are independent. Further note that for $x$ such that $x_i=x'$,
$$\Pr[\Pi=\tau\mid X=x, X_i=x', M=m, Z=z] =\Pr[\Pi(x)=\tau]~.$$

By Lemma~\ref{lem:rectang}, there exist functions $q_i$ and $p_{-i}$ such that 
$$\forall~x\in\{0,1\}^k,~\Pr[\Pi(x)=\tau]=q_i(x_i,\tau_i)p_{-i}(x_{-i},\tau)~.$$

Therefore we can write
\begin{align*}
\Pr[\Pi=\tau\mid X_i=x', M=m, Z=z] &= 
\sum\limits_{x\in\{0,1\}^k}\left(\delta_{x_i,x'}q_i(x_i,\tau_i)p_{-i}(x_{-i},\tau)\times
\Pr[X_{-i}=x_{-i}\mid M=m, Z=z]\right)\\
&~~\\
&=q_i(x',\tau_i)\sum\limits_{\hat{x}\in\{0,1\}^{k-1}}\left(p_{-i}(\hat{x},\tau)\times
\Pr[X_{-i}=\hat{x}\mid M=m, Z=z]\right)\\
&~~\\
&=q_i(x',\tau_i)c(m,z,\tau)~,
\end{align*}

where 
$c(m,z,\tau)=\sum\limits_{\hat{x}\in\{0,1\}^{k-1}}p_{-i}(\hat{x},\tau)
\Pr[X_{-i}=\hat{x}\mid M=m, Z=z]$.

\end{proof}

\begin{proof}[Proof of Lemma~\ref{lem:diagspe}]
\adifuture{go over correctness once more}
Using Lemma~\ref{lem:rectangspe}, we write \\
$\Pr[\Pi_i[0,0,j]=\overline{\tau}] = q_i(0,\overline{\tau})c_i(0,j,\overline{\tau})$ and
$\Pr[\Pi_i[1,1,j]=\overline{\tau}] = q_i(1,\overline{\tau})c_i(1,j,\overline{\tau})$.\\

Using Lemma~\ref{lem:rectang}, we write\\
$\Pr[\Pi_i(\overline{e}^k_{i,j})=\overline{\tau}]=q_i(0,\overline{\tau})q_{-i}(\overline{e}^{k-1}_j,\overline{\tau})$ and 
$\Pr[\Pi_i(\overline{e}^k_j)=\overline{\tau}]=q_i(1,\overline{\tau})q_{-i}(\overline{e}^{k-1}_j,\overline{\tau})$.\\

Note that $\Pi_i[1,1,j]=\Pi_i(\overline{e}^k_j)$, and thus 
$q_i(1,\overline{\tau}) \neq 0 \Rightarrow c_i(1,j,\overline{\tau}) = q_{-i}(\overline{e}^{k-1}_j,\overline{\tau})$.

By Proposition~\ref{prop:Hellsquare},
\begin{align*}
1-h^2(\Pi_i[0,0,j],\Pi_i[1,1,j])&=
\sum\limits_{\overline{\tau}}\sqrt{\Pr[\Pi_i[0,0,j]=\overline{\tau}]\Pr[\Pi_i[1,1,j]=\overline{\tau}]}\\
&=\sum\limits_{\overline{\tau}}\sqrt{q_i(0,\overline{\tau})c_i(0,j,\overline{\tau})q_i(1,\overline{\tau})c_i(1,j,\overline{\tau})}\\
&\le\sum\limits_{\overline{\tau}}\sqrt{q_i(0,\overline{\tau})q_i(1,\overline{\tau})}
\left(\frac{c_i(0,j,\overline{\tau})+c_i(1,j,\overline{\tau})}{2}\right)\\
&\le\frac{1}{2}\left(\sum\limits_{\overline{\tau}}\sqrt{q_i(0,\overline{\tau})c_i(0,j,\overline{\tau})q_i(1,\overline{\tau})c_i(0,j,\overline{\tau})}\right.~+\\
&~~~~~~~~~\left.\sum\limits_{\overline{\tau}}\sqrt{
q_i(0,\overline{\tau})c_i(1,j,\overline{\tau})q_i(1,\overline{\tau})c_i(1,j,\overline{\tau})}\right)\\
&\le\frac{1}{2}\left(\sum\limits_{\overline{\tau}}\sqrt{q_i(0,\overline{\tau})c_i(0,j,\overline{\tau})q_i(1,\overline{\tau})c_i(0,j,\overline{\tau})}\right.~+\\
&~~~~~~~~~\left.\sum\limits_{\overline{\tau}\mid q_i(1,\overline{\tau})\neq 0}\sqrt{
q_i(0,\overline{\tau})c_i(1,j,\overline{\tau})q_i(1,\overline{\tau})c_i(1,j,\overline{\tau})}\right)\\
&\le\frac{1}{2}\left(\sum\limits_{\overline{\tau}}\sqrt{q_i(0,\overline{\tau})c_i(0,j,\overline{\tau})q_i(1,\overline{\tau})c_i(0,j,\overline{\tau})}\right.~+\\
&~~~~~~~~~\left.\sum\limits_{\overline{\tau}\mid q_i(1,\overline{\tau})\neq 0}\sqrt{
q_i(0,\overline{\tau})q_{-i}(\overline{e}^{k-1}_j,\overline{\tau})q_i(1,\overline{\tau})q_{-i}(\overline{e}^{k-1}_j,\overline{\tau})}\right)\\
&\le\frac{1}{2}\left(\sum\limits_{\overline{\tau}}\sqrt{q_i(0,\overline{\tau})c_i(0,j,\overline{\tau})q_i(1,\overline{\tau})c_i(0,j,\overline{\tau})}\right.~+\\
&~~~~~~~~~\left.\sum\limits_{\overline{\tau}}\sqrt{
q_i(0,\overline{\tau})q_{-i}(\overline{e}^{k-1}_j,\overline{\tau})q_i(1,\overline{\tau})q_{-i}(\overline{e}^{k-1}_j,\overline{\tau})}\right)\\
&\le\frac{1}{2}\left(\sum\limits_{\overline{\tau}}\sqrt{\Pr[\Pi_i[0,0,j]=\overline{\tau}]\Pr[\Pi_i[1,0,j]=\overline{\tau}]}\right.~+\\
&~~~~~~~~~\left.\sum\limits_{\overline{\tau}}\sqrt{\Pr[\Pi_i(\overline{e}^k_{i,j})=\overline{\tau}]\Pr[\Pi_i(\overline{e}^k_j)=\overline{\tau}]}\right)\\
&\le\frac{1}{2}(1-h^2(\Pi_i[0,0,j],\Pi_i[1,0,j]) + 1-h^2(\Pi_i(\overline{e}^k_{i,j}),\Pi_i(\overline{e}^k_j)))\\
&\le 1-h^2(\Pi_i(\overline{e}^k_{i,j}),(\Pi_i(\overline{e}^k_j))~.
\end{align*}
\end{proof}

\begin{proof}[Proof of Lemma~\ref{lem:localspe}]
\adifuture{go over correctness once more}
Using Lemma~\ref{lem:rectang}, we write
$$\Pr[\Pi_{i}(\overline{e}^k_{i,j})=\overline{\tau}]= q_i(0,\overline{\tau})q_{-i}(\overline{e}^{k-1}_j,\overline{\tau})$$ and
$$\Pr[\Pi(\overline{e}^k_{i,j})=\tau]= q_i(0,\tau_i)p_{-i}(\overline{e}^{k-1}_j,\tau)~.$$ As
$\Pr[\Pi_{i}(\overline{e}^k_{i,j})=\overline{\tau}] = 
\sum\limits_{\tau\mid\tau_i=\overline{\tau}} \Pr[\Pi(\overline{e}^k_{i,j})=\tau]$ we have \\
$$q_i(0,\overline{\tau})q_{-i}(\overline{e}^{k-1}_j,\overline{\tau}) = 
\sum\limits_{\tau\mid\tau_i=\overline{\tau}}q_i(0,\tau_i)p_{-i}(\overline{e}^{k-1}_j,\tau)
=q_i(0,\overline{\tau})\sum\limits_{\tau\mid\tau_i=\overline{\tau}}p_{-i}(\overline{e}^{k-1}_j,\tau)~,$$ and thus 
$$q_i(0,\overline{\tau})\neq 0 \Rightarrow 
q_{-i}(\overline{e}^{k-1}_j,\overline{\tau}) = \sum\limits_{\tau\mid\tau_i=\overline{\tau}}p_{-i}(\overline{e}^{k-1}_j,\tau)~.$$

Using Proposition~\ref{prop:Hellsquare}, we can write
\begin{align*}
1-h^2(\Pi_{i}(\overline{e}^k_{i,j}),\Pi_{i}(\overline{e}^k_{j}))&=
\sum\limits_{\overline{\tau}}\sqrt{\Pr[\Pi_{i}(\overline{e}^k_{i,j})=\overline{\tau}]\Pr[\Pi_{i}(\overline{e}^k_{j})=\overline{\tau}]}\\
&=\sum\limits_{\overline{\tau}}\sqrt{
q_i(0,\overline{\tau})q_{-i}(\overline{e}^{k-1}_j,\overline{\tau})q_i(1,\overline{\tau})q_{-i}(\overline{e}^{k-1}_j,\overline{\tau})}\\
&=\sum\limits_{\overline{\tau}}\sqrt{
q_i(0,\overline{\tau})q_i(1,\overline{\tau})}q_{-i}(\overline{e}^{k-1}_j,\overline{\tau})\\
&=\sum\limits_{\overline{\tau}\mid 	q_i(0,\overline{\tau})\neq 0}\sqrt{
q_i(0,\overline{\tau})q_i(1,\overline{\tau})}q_{-i}(\overline{e}^{k-1}_j,\overline{\tau})\\
&=\sum\limits_{\overline{\tau}\mid 	q_i(0,\overline{\tau})\neq 0}\left(\sqrt{
q_i(0,\overline{\tau})q_i(1,\overline{\tau})}\sum\limits_{\tau\mid\tau_i=\overline{\tau}}p_{-i}(\overline{e}^{k-1}_j,\tau)\right)\\
&=\sum\limits_{\overline{\tau}}\left(\sqrt{
q_i(0,\overline{\tau})q_i(1,\overline{\tau})}\sum\limits_{\tau\mid\tau_i=\overline{\tau}}p_{-i}(\overline{e}^{k-1}_j,\tau)\right)\\
&=\sum\limits_{\tau}\left(\sqrt{
q_i(0,\tau_i)q_i(1,\tau_i)}p_{-i}(\overline{e}^{k-1}_j,\tau)\right)\\
&=\sum\limits_{\tau}\sqrt{\Pr[\Pi(\overline{e}^k_{i,j})=\tau]\Pr[\Pi(\overline{e}^k_{j})=\tau]}\\
&=1-h^2(\Pi(\overline{e}^k_{i,j}),\Pi(\overline{e}^k_{j}))~.
\end{align*}
\end{proof}

\begin{proof}[Proof of Lemma~\ref{le:ic_hat}]

For the purpose of the proof we define a certain  order on the messages in $\Pi_i$, i.e., on all messages sent and received by player $i$ as follows.
We order the messages of $\Pi_i$ by (local) rounds of player $i$, and inside each round have first the messages sent by 
player $i$, ordered by the index of the recipient, and have then the messages received by player $i$, ordered by the 
index of the sender.  For a given player $i$, we denote the sequence thus defined as $(B^d)_{d\geq0}$.

Now, by the chain rule, applied on the messages of $\Pi_i$ by the order we just defined, and after rearranging the summands, we have 

$$I(X_i ; \Pi_i \mid R^p M Z) = \sum\limits_{\ell} I(X_i ; \M{i}{\ell}{s} \mid \T{i}{\ell}{s} R^p M Z) +  
\sum\limits_{\ell} I(X_i ; \M{i}{\ell}{r} \mid \T{i}{\ell}{r} R^p M Z)~.$$

\noindent We now show that every summand of the second sum equals $0$.\\
To this end, we  now show by induction on the index $d$ that  $\forall d, ~I(X_i;X_{-i}\mid M Z R^p B^0\dots B^d)=0$.
We have {$I(X_i;X_{-i}\mid M Z R^p)=0$}, because according to $\mu^n$, conditioned on $MZ$, $X_i$ and $X_{-i}$
 are independent.
 Assume now the induction hypothesis  that for some $d$, ${I(X_i;X_{-i}\mid M Z R^p B^0\dots B^d)=0}$. 
 If the message $B^{d+1}$ is {\em sent} by player $i$, then $B^{d+1}$ is 
 a function of $X_i$, $R^p$ and $B^0\dots B^{d}$ and thus
\begin{align*}
I(X_i;X_{-i}\mid M Z R^p B^0\dots B^{d+1}) &= 
H(X_{-i}\mid M Z R^p B^0\dots B^{d+1})~- \\
&~~~~~H(X_{-i}\mid M Z R^p B^0\dots B^{d+1} X_i)\\
&\le H(X_{-i}\mid M Z R^p B^0\dots B^{d})~- \\
&~~~~~H(X_{-i}\mid M Z R^p B^0\dots B^{d} X_i)\\
&= I(X_i;X_{-i}\mid M Z R^p B^0\dots B^{d})\\
&=0~.
\end{align*}
Similarly, if the message $B^{d+1}$ is {\em received} by player $i$, then $B^{d+1}$ is a
 function of $X_{-i}$, $R^p$ and $B^0\dots B^{d}$ and thus
\begin{align*}
I(X_i;X_{-i}\mid M Z R^p B^0\dots B^{d+1}) &= 
H(X_i\mid M Z R^p B^0\dots B^{d+1})~-\\
&~~~~~H(X_i\mid M Z R^p B^0\dots B^{d+1} X_{-i})\\
&\le H(X_i\mid M Z R^p B^0\dots B^{d})~-\\
&~~~~~H(X_i\mid M Z R^p B^0\dots B^{d} X_{-i})\\
&= I(X_i;X_{-i}\mid M Z R^p B^0\dots B^{d} )\\
&=0~.
\end{align*}
Thus we have that 
\begin{equation}
\forall d, ~~I(X_i;X_{-i}\mid M Z R^p B^0\dots B^d)=0~.
\label{eq:Delta}
\end{equation}
From Eq.~\eqref{eq:Delta},
by choosing the relevant $d$ for any given $\ell$, we  can also write  for all $\ell$  \\
$I(X_i ; X_{-i} \mid \T{i}{\ell}{r} R^p M Z)=0$. 
Applying Lemma \ref{lem:tec1} with $A=X_i$, $B=X_{-i}$, $C=(\T{i}{\ell}{r},R^p)$, $D=(M,Z)$ and 
$\phi=\M{i}{\ell}{r}=\varphi(\T{i}{\ell}{r},R^p,X_{-i})$ yields $I(X_i ; \M{i}{\ell}{r} \mid \T{i}{\ell}{r} R^p M Z)=0$. We have
 thus shown that 
$$I(X_i ; \Pi_i \mid R^p M Z) = \sum\limits_\ell I(X_i ; \M{i}{\ell}{s} \mid \T{i}{\ell}{s} R^p M Z)~,$$
and thus 
\begin{equation}
\label{eq:to_relate_to_hat}
\sum\limits_{i=1}^k I(X_i ; \Pi_i \mid R^p M Z)  = \sum\limits_i \sum\limits_\ell I(X_i ; \M{i}{\ell}{s} \mid \T{i}{\ell}{s} R^p M Z)~.
\end{equation}

We note that the equation $I(X_i ; \Pi_i \mid R^p M Z) = \sum\limits_\ell I(X_i ; \M{i}{\ell}{s} \mid \T{i}{\ell}{s} R^p M Z)$ that we
 proved above formalizes the intuitive  assertion that if we consider the messages of $\Pi_i$ in their order of appearance, 
 then additional information on $X_i$ is obtained only from messages {\em sent by}  player $i$, but not from messages
 {\em received by} player $i$.

We now relate  $\widehat{\IC}$  to the right-hand-side of Eq.~\eqref{eq:to_relate_to_hat}.
Starting from the definition of  $\widehat{\IC}$ and using the chain rule, 
we  decompose $\widehat{\IC}$ into a sum over all messages received in the protocol:  
\begin{align*}
\widehat{\IC}(\pi) &= 
 \sum\limits_{j=1}^{k}\sum\limits_{\ell' \geq 0} I(X_{-j} ; \M{j}{\ell'}{r} \mid \M{j}{0}{r}\ldots \M{j}{\ell'-1}{r}X_{j} R^p M Z) \\
&= \sum\limits_{j=1}^{k}\sum\limits_{\ell' \geq 0} I(X_{-j} ; \M{j}{\ell'}{r} \mid \T{j}{\ell'}{r} X_{j} R^p M Z)~,
\end{align*}
where the second equality follows from the fact that the messages in $\T{j}{\ell'}{r}$ which are {\em sent} by player 
$j$ are  a function of $X_j$, $R^p$ and of the messages in $\T{j}{\ell'}{r}$ which are {\em received}  by player  $j$.

We now rearrange this sum by considering the messages from the point of view of the sender rather than the receiver.
In what follows we use $j$ as a shorthand of $j(i,\ell)$ and $\ell'$ as a shorthand of $\ell'(i,\ell)$.\footnote{Recall
that $j(i,\ell)$ and $\ell'(i,\ell)$ are defined such that message $\M{i}{\ell}{s}$ is identified with the message $\M{j(i,\ell)}{\ell'(i,\ell)}{r}$.}
We have
$$\widehat{\IC}(\pi)=\sum\limits_{i=1}^k\sum\limits_{\ell \geq 0} I(X_{-{j}} ; \M{i}{\ell}{s} \mid \T{j}{\ell'}{r} X_{j} R^p M Z)~.$$

To conclude the proof our objective  now is to show that for any message $\M{i}{\ell}{s}$, 
\begin{equation}
\label{eq:to_prove}
I(X_i ; \M{i}{\ell}{s} \mid \T{i}{\ell}{s} R^p M Z) \le I(X_{-{j}} ; \M{i}{\ell}{s} \mid \T{j}{\ell'}{r} X_{j} R^p M Z)~.
\end{equation}
Observe  that since $\M{i}{\ell}{s}$ is determined by  $X_iR^p\T{i}{\ell}{s}$,  we have  
$H(\M{i}{\ell}{s} \mid X_i \T{i}{\ell}{s}  R^p M Z) = 0$, and thus\\
${I(X_i ; \M{i}{\ell}{s} \mid \T{i}{\ell}{s}  R^p M Z) = H(\M{i}{\ell}{s} \mid \T{i}{\ell}{s}  R^p M Z)}$.
Similarly, we have that \\
$I(X_{-{j}} ;\M{i}{\ell}{s} \mid \T{j}{\ell'}{r} X_{j} R^p M Z) = H(\M{i}{\ell}{s} \mid \T{j}{\ell'}{r} X_{j}  R^p M Z)$.
Thus,
\begin{align*}
I(X_i ; \M{i}{\ell}{s} \mid \T{i}{\ell}{s} R^p M Z) &\le I(X_{-{j}} ; \M{i}{\ell}{s} \mid \T{j}{\ell'}{r} X_{j} R^p M Z)\\
&\Updownarrow\\
H(\M{i}{\ell}{s} \mid \T{i}{\ell}{s} R^p M Z) &\le H(\M{i}{\ell}{s} \mid \T{j}{\ell'}{r} X_{j} R^p M Z)\\
&\Updownarrow\\
I(\M{i}{\ell}{s} ; \T{i}{\ell}{s} R^p M Z) &\ge I(\M{i}{\ell}{s} ; \T{j}{\ell'}{r} X_{j} R^p M Z)~.
\end{align*}

The last inequality clearly holds if $I(\M{i}{l}{s} ; \T{i}{\ell}{s} R^p M Z) = I(\M{i}{\ell}{s} ; \T{i}{\ell}{s} \T{j}{\ell'}{r} X_{j} R^p M Z)$, 
 which itself holds if  
\begin{equation}
\label{eqn:mess}
I(\M{i}{\ell}{s} ; \T{j}{\ell'}{r} X_{j} \mid \T{i}{\ell}{s} R^p M Z) = 0~.
\end{equation}

Notice that given the value of $\T{i}{l}{s} R^p M Z$, $\M{i}{\ell}{s}$ is determined by $X_i$, and thus by the data processing inequality 
(Proposition~\ref{prop:dataproc}) we have
$$I(X_i ; \T{j}{\ell'}{r} X_{j} \mid \T{i}{\ell}{s} R^p M Z) \ge I(\M{i}{\ell}{s} ; \T{j}{\ell'}{r} X_{j} \mid \T{i}{\ell}{s} R^p M Z)~,$$
and Eq.~\ref{eqn:mess} holds if  $I(X_i ; \T{j}{\ell'}{r} X_{j} \mid \T{i}{\ell}{s} R^p M Z) = 0$.

 Let $t$ be the (global) round in which 
player $j$ receives $\M{j}{\ell'}{r}$, which is also the (global) round in which player $i$ sends $\M{i}{\ell}{s}$ (recall that 
 in fact  $\M{j}{\ell'}{r}$ is the same message as $\M{i}{\ell}{s}$).
Now,  all the messages in 
$\T{j}{\ell'}{r}$ are received or sent by player $j$ no later than  (global) round $t$,
 and all the messages sent by player $i$ which are not in $\T{i}{\ell}{s}$ are sent by player $i$ no earlier than (global) round $t$.
 Hence,  $\T{j}{\ell'}{r}$ is a function of $(X_{-i},\T{i}{\ell}{s})$. We also have trivially that $X_j$ is a function of $X_{-i}$.
 Thus, 
the data processing inequality implies that 
$$I(X_i ; \T{j}{\ell'}{r} X_{j} \mid \T{i}{\ell}{s} R^p M Z) \le I(X_i ; X_{-i} \T{i}{\ell}{s} \mid \T{i}{\ell}{s} R^p M Z) =
 I(X_i ; X_{-i} \mid \T{i}{\ell}{s} R^p M Z)~.$$ 
By Eq.~\eqref{eq:Delta}, 
 $I(X_i ; X_{-i} \mid \T{i}{\ell}{s} R^p M Z) = 0$, which  concludes the proof of Eq.~\eqref{eqn:mess}
 and therefore  of Eq.~\eqref{eq:to_prove} and of the lemma.
\end{proof}

\LongVersionEnd

\end{document}